\documentclass[sigconf]{acmart}

\usepackage{graphicx} 
\usepackage{hyperref}
\usepackage{subcaption}
\usepackage{amsmath}
\usepackage{comment}

\usepackage{amssymb}
\usepackage{float}
\usepackage{cleveref}
\usepackage{dblfloatfix}
\usepackage{multirow}
\usepackage{geometry} 
\usepackage{caption} 
\usepackage{colortbl} 
\usepackage{xcolor} 
\usepackage{array}
\usepackage{booktabs} 
\usepackage[subtle]{savetrees}
\usepackage{enumitem}
\usepackage{colortbl} 
\AtBeginDocument{%
  }

\setcopyright{acmlicensed}
\copyrightyear{2018}
\acmYear{2018}
\acmDOI{XXXXXXX.XXXXXXX}
\acmConference[Conference acronym 'XX]{Make sure to enter the correct
  conference title from your rights confirmation emai}{June 03--05,
  2018}{Woodstock, NY}
\acmISBN{978-1-4503-XXXX-X/18/06}

\begin{document}


\title{On the Role of Weight Decay in \\ Collaborative Filtering: A Popularity Perspective}

\author{Donald Loveland}
\affiliation{
  \institution{University of Michigan, Ann Arbor}
  \city{Ann Arbor}
  \country{USA}
}
\email{dlovelan@umich.edu}

\author{Mingxuan Ju}
\affiliation{
  \institution{Snap Inc.}
  \city{Bellevue}
  \country{USA}
}
\email{mju@snap.com}

\author{Tong Zhao}
\affiliation{
  \institution{Snap Inc.}
  \city{Bellevue}
  \country{USA}
}
\email{tong@snap.com}

\author{Neil Shah}
\affiliation{
  \institution{Snap Inc.}
  \city{Bellevue}
  \country{USA}
}
\email{nshah@snap.com}

\author{Danai Koutra}
\affiliation{
  \institution{University of Michigan, Ann Arbor}
  \city{Ann Arbor}
  \country{USA}
}
\email{dkoutra@umich.edu}

\renewcommand{\shortauthors}{Donald Loveland et al.}

\begin{abstract}
  Collaborative filtering (CF) enables large-scale recommendation systems by encoding information from historical user-item interactions into dense ID-embedding tables. However, as embedding tables grow, closed-form solutions become impractical, often necessitating the use of mini-batch gradient descent for training. Despite extensive work on designing loss functions to train CF models, we argue that one core component of these pipelines is heavily overlooked: \textit{weight decay}. Attaining high-performing models typically requires careful tuning of weight decay, regardless of loss, yet its necessity is not well understood. In this work, we question \textit{why} weight decay is crucial in CF pipelines and \textit{how} it impacts training. Through theoretical and empirical analysis, we surprisingly uncover that weight decay's primary function is to encode popularity information into the magnitudes of the embedding vectors. Moreover, we find that tuning weight decay acts as a coarse, non-linear, knob to influence preference towards popular or unpopular items.
  Based on these findings, we propose \textbf{PRISM} (\textbf{P}opularity-awa\textbf{R}e \textbf{I}nitialization \textbf{S}trategy for embedding \textbf{M}agnitudes), a straightforward yet effective solution to simplify the training of high-performing CF models. PRISM pre-encodes the popularity information typically learned through weight decay, eliminating its necessity. Our experiments show that PRISM improves performance by up to 4.77\% and reduces training times by 38.48\%, compared to state-of-the-art training strategies. Additionally, we parameterize PRISM to modulate the initialization strength, offering a cost-effective and meaningful strategy to mitigate popularity bias. 
\end{abstract}


\received{20 February 2007}
\received[revised]{12 March 2009}
\received[accepted]{5 June 2009}

\maketitle

\section{Introduction}

The vast amount of online data poses significant challenges for content navigation and discovery \cite{cooley_web_1997}. In response, recommender systems have become essential to distill content into a manageable set of personalized recommendations \cite{gomez2016netflix, yang2023newsrec, dou206surveycf_social, sankar2021friend, oord2013music}. At the core of many recommender system is collaborative filtering (CF), a technique that leverages historical user-item preferences to generate new recommendations \cite{su2009cfsurveya, chen2018cfsurveyb}. Among the most prominent CF methods is matrix factorization (MF), which translates sparse historical interaction data into dense embeddings by learning low-rank user and item matrices \cite{koren2009mfrec,wang2022towards, xue2017deeprec}. Recent advancements in MF have also considered incorporating non-linearities through multilayer perceptrons (MLPs) \cite{he2017ncf, covington2016youtube, cheng2016widedeeplearning}, as well as message-passing, to enable the encoding of richer user-item interactions \cite{he2020lightgcn, wang2019neural}. 

When training CF models, closed-form solutions become impractical with large datasets \cite{bokde2015mf_svd}, leading to the use of iterative learning algorithms, such as mini-batch gradient descent \cite{rendle2009bpr, wu2022ssm}. However, despite their empirical success, gradient-based methods cannot guarantee an optimal solution \cite{peng2019convexmf}. Thus, significant effort has been devoted to designing losses that improve the optimization of these approaches \cite{wang2022towards}. One core design consideration is deciding which geometric properties of the embedding vectors to optimize. Early loss functions, such as Bayesian Personalized Ranking (BPR) \cite{rendle2009bpr}, leverage the dot product similarity between interacted pairs, promoting the learning of both magnitude and angle information. In contrast, more recent loss functions have shifted to prioritize meaningful angles, e.g. through the use of cosine similarity \cite{wu2022ssm} or normalized embeddings \cite{chen2024mitigate, wang2022towards}. We refer to these as \textit{angle-based losses} given the loss values are not impacted by embedding magnitude. These losses have shown exceptional model performance, leading to a continued prioritization of angle-based optimization \cite{park2023toward, lin2024recommendationmodelsamplifypopularity}.

Despite angle-based losses superior performance, a notable inconsistency frequently arises in their evaluation -- the application of \textit{weight decay}. Most angle-based losses recommend careful tuning of weight decay \cite{wu2022ssm, wang2022towards, park2023toward}, yet they lack justification for its necessity given weight decay regularizes the embedding \textit{magnitudes}. Even more surprising, angle-based losses display high sensitivity to weight decay strength, leading to significant performance degradation if not properly tuned (shown in \Cref{fig:motivation_dau_bpr_wd}). This seemingly contradictory yet necessary application of weight decay on losses that do not explicitly leverage magnitude information leads us to pose the question:
\begin{center}
    \textbf{\textit{Why is weight decay crucial to learning high-performing CF models, and what properties does it encode?}}
\end{center}
Understanding how weight decay influences training allows us to discern the roles of magnitude and angle information. However, practically, the process of tuning weight decay becomes infeasible as the datasets grow, potentially leading to poor performance if the optimal strength is not found. Moreover, incorporating weight decay can be costly given it requires frequent parameter updates to large embedding tables which may not be loaded in memory. Thus, this inconsistency between angle-based losses and weight decay not only highlights the importance of deeply understanding weight decay's role, but also motivates the exploration of alternative strategies capable of efficiently encoding similar properties.

We begin by systematically analyzing state-of-the-art (SOTA) loss functions to evaluate the properties they encode into the embedding matrices. Although many of these losses are angle-based, we unexpectedly discover that the magnitudes of the embedding vectors capture the underlying popularity information, a behavior more typically associated with losses that use both magnitude and angle. Through comprehensive analysis, we show that weight decay is the direct cause of this phenomenon, resulting from imbalanced gradient updates across parameters due to mini-batch gradient descent.
Surprisingly, we find that this magnitude encoding is crucial for angle-based loss functions to achieve SOTA performance; without the popularity encoding from weight decay, performance significantly regresses. In fact, more complex losses (e.g. \texttt{DirectAU}), are often out-performed by simple baselines (e.g. \texttt{BPR}) when weight decay is not applied, raising important questions around what designs are necessary for achieving high performance. Thus, our first contribution is the insight that weight decay is essential for angle-based losses to achieve SOTA through the encoding of popularity information, and angle information is only partially responsible for their high performances.

Building on our analysis, we question whether weight decay is necessary for achieving high performance. As angle-based losses focus on angle information, while weight decay encodes magnitude information correlated with popularity, we introduce a cost-effective initialization strategy called \textbf{PRISM}: a \textbf{P}opularity-awa\textbf{R}e \textbf{I}nitialization \textbf{S}trategy for the embedding \textbf{M}agnitudes. PRISM splits the learning of angle and magnitude by directly integrating popularity information into the embedding magnitudes, allowing training to focus on learning angle information. By replacing weight decay with PRISM, \textit{our models achieve similar performance to those trained with weight decay}, without the need for costly hyperparameter tuning. Additionally, we find that \textit{models trained with PRISM converge on average 38\% faster than those reliant on weight decay}, as the latter requires a considerable number of epochs to encode magnitude information.
We additionally incorporate an encoding strength parameter into PRISM that enables performance equal to, or better than, the best weight decay-tuned models, with significant performance improvements on less popular items. 
Collectively, our contributions provide a new fundamental understanding of CF training and offer insights into the role of weight decay in this process. Our key contributions are below:

\begin{enumerate}[leftmargin=*]
    \item \textbf{Linking Weight Decay to Popularity:} We establish the necessary role of weight decay in CF by theoretically and empirically demonstrating its link to encoding popularity. This insight challenges current understanding of SOTA performance, as many high-performing training strategies fail without weight decay.  
    \item \textbf{Popularity-aware Initialization for Embedding Magnitudes: } We propose PRISM, a magnitude-based embedding initialization strategy that is efficient and easy to implement, requiring a single computation at the onset of training. PRISM is able to eliminate the need for weight decay without performance loss. 
    \item \textbf{Practical Benefits of Removing Weight Decay: } Through our parameterization of PRISM, we demonstrate that if one chooses to hyperparameter tune the initialization strategy, there are significant potential benefits, including (a) improved model performance, (b) accelerated training times, and (c) effective mitigation of popularity bias. This tuning capability highlights the practical utility and flexibility of PRISM in real-world scenarios.
\end{enumerate}

\section{Preliminaries and Related Work}

In this section, we offer an overview of CF, with a focus on model backbones and loss function designs. We also formalize the geometric properties pertinent to user and item embeddings.

\subsection{Collaborative Filtering}

Collaborative filtering leverages historical interaction data, $E$, to learn user and item embeddings. Given a set of $n$ users, $U$, and a set of $m$ items, $I$,  the interaction matrix is denoted as $\mathbf{E} \in \mathbb{Z}^{n \times m}$. User and item embeddings are represented by $\mathbf{U} \in \mathbb{R}^{n \times d}$ and $\mathbf{I} \in \mathbb{R}^{m \times d}$, where $d$ represents the dimensionality of the embedding vectors. MF is a prevalent approach to learning these embeddings due to its ability to capture latent factors that influence interactions. Specifically, for a user $u$'s embedding $\mathbf{u}$ and an item $i$'s embedding $\mathbf{i}$, MF aims to encourage that $\mathbf{E}_{u, i} \approx \mathbf{u} \cdot \mathbf{i}^\top$. Despite MF possessing closed form solutions \cite{bokde2015mf_svd}, in practice, MF is trained with mini-batch gradient descent given the entire interaction matrix generally cannot fit into memory. Thus, significant effort has been invested in improving these gradient-based training methods.

Recent advancements have extended MF by incorporating different embedding functions.
The most prominent approaches revolve around utilizing deep neural networks (DNNs) to introduce non-linearities into the interaction calculations \cite{he2017ncf}. As an example, user and item matrices can be processed through DNNs $ F $ and $ G $, transforming the interaction calculation such that  $\mathbf{E} \approx F(\mathbf{U}) \cdot G(\mathbf{I})^\top$.  Additionally, some approaches explore replacing the dot product with a learnable function, such as using a DNN, $H$, to compute $\mathbf{E}_{u, i} \approx H(\mathbf{u} || \mathbf{i})$ where $||$ denotes concatenation. \cite{wang2021dcnv2, cheng2016widedeeplearning, zhang2019dlrecsurvey}. Graph-based CF has also been proposed, utilizing message passing to propagate user and item information \cite{he2020lightgcn, gao2023gnnsurvey, wang2019neural}. While MF and its neural network variants rely on their loss functions to determine which geometric properties of the embedding vectors are optimized during training, message-passing architectures have been shown to inherently encode magnitude information \cite{wu2022ssm}. As such, we explore how different losses prioritize magnitude versus angle information during training.

\begin{figure*}[t!] 
    \centering \includegraphics[width=17.cm]{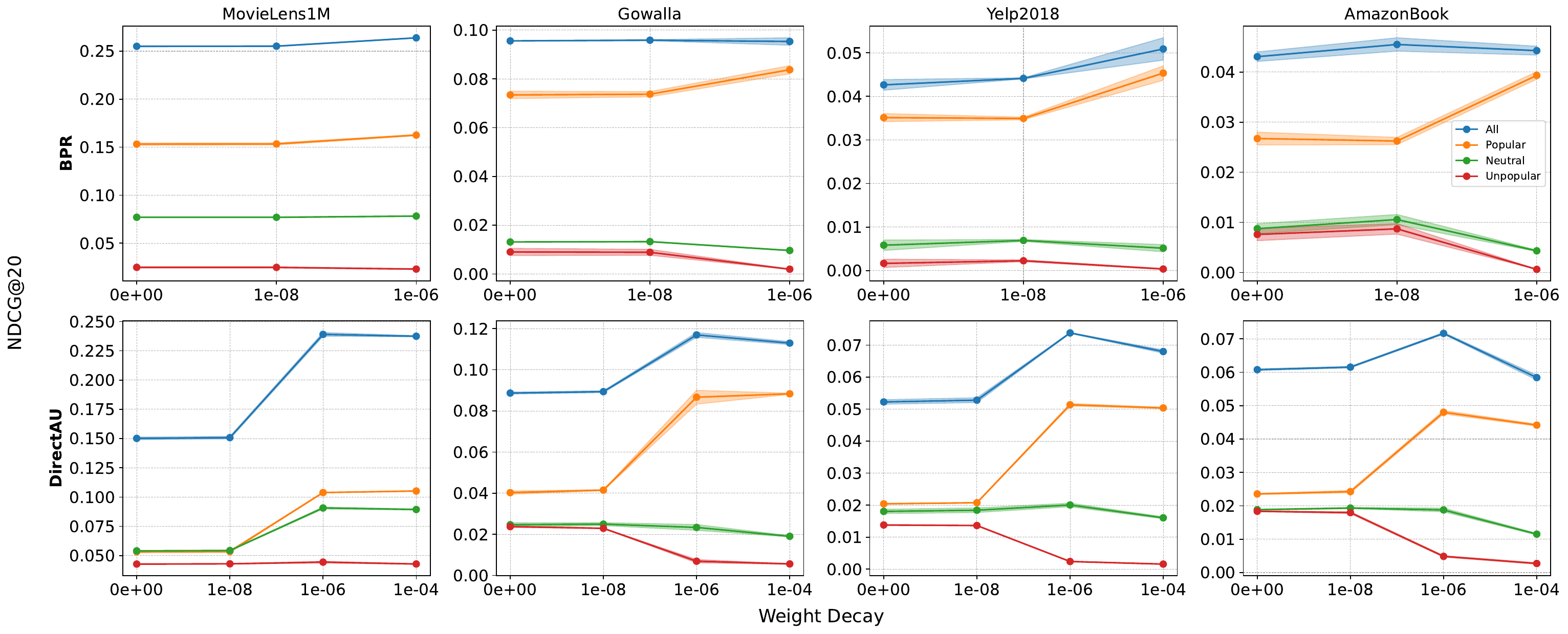} 
    \vspace{-0.3cm}
    \caption{\textbf{Performance with Varying Weight Decay: BPR (top) and DirectAU (bottom)}. Increasing the strength of weight decay demonstrates a preference towards items with high popularity. Weight decay strength of $10^{-4}$ is omitted for BPR as it causes performance to approach 0 due to over-regularization. Filled regions denote standard deviation over 3 runs.} 
\label{fig:motivation_dau_bpr_wd} 
\end{figure*}

\subsection{Embedding Vector Geometry}

The design of loss functions is crucial in guiding the gradient descent process to encode meaningful information. Although losses can leverage various mechanisms, we focus primarily on whether a loss emphasizes encoding magnitude or angle information. Below, we formalize the notion of an \textit{angle-based loss}.

\begin{definition}
\textbf{\textit{Angle-based Loss}}. For embedding vectors $\mathbf{u}$ and $\mathbf{i}$, a loss function $L(\mathbf{u}, \mathbf{i})$ is an \textit{angle-based loss} if it satisfies the condition $L(\mathbf{u}, \mathbf{i}) = L(c_{1}\mathbf{u} , c_{2}\mathbf{i})$,
where $c_{1}$ and $c_{2}$ are arbitrary scalars. 
\label{def:angle_loss}
\end{definition}
Given \Cref{def:angle_loss}, an angle-based loss does not depend on the magnitudes of the embedding vectors, focusing instead on their relative angles. For example, the cosine similarity loss remains unchanged regardless of vector magnitudes.  Similarly, computing the Euclidean distance between normalized vectors, as in DirectAU or MAWU, also ignores magnitudes. In contrast, losses without this constraint, like BPR and its dot product similarity metric, can learn both magnitude and angle information. Exact formulations of these losses are given in \Cref{section:app_loss_funcs}. It is important to note that our definition of angle-based loss pertains only to the loss computation itself. Consequently, the \textit{gradient descent process using an angle-based loss may still alter both magnitude and angle information} of the embeddings. We delve into this distinction in \Cref{section:theory}, explaining how angle-based losses encode popularity information, leveraging this property.

\subsubsection{What do these properties encode?}

Understanding the importance of magnitude versus angle information in modern CF pipelines has gained more attention, but is still limited in scope. One investigation into this distinction is seen in the SSM loss, which emphasizes the benefits of prioritizing angle information \cite{wu2022ssm}. However, the strong empirical performance of SSM heavily relies on a mechanism for learning magnitude, achieved through message passing. Other losses, such as DirectAU and MAWU, also emphasize angle importance by promoting uniformity of non-interacted samples on the hypersphere  \cite{wang2022towards, park2023toward}. Despite this, these approaches employ weight decay during training without addressing the role of magnitude information in their performance.
Subsequent work explicitly links magnitude information to popularity by examining the magnitude distribution in trained models \cite{Chen_2023}. Given this relationship, many prevalent post-hoc popularity de-biasing methods utilize re-weighting strategies based on item popularity \cite{Chen_2021, zhu2021popbiascf}. These methods apply weighting to item scores, thereby implicitly modulating the magnitude of embedding vectors. Similar mitigation strategies are observed in graph-based collaborative filtering, where the strength of magnitude normalization is adjusted \cite{zhao2022acc_nov}.

Despite efforts to remove popularity information from learned embeddings through altering magnitudes, it is important to note that popularity is not inherently negative. Popular items are typically favored by a large user base, and effectively leveraging these popularity signals can enhance model performance \cite{Klimashevskaia_2024, lin2024recommendationmodelsamplifypopularity}. Given this perspective, we concentrate on establishing \textit{how} embedding magnitudes encode popularity information, particularly in angle-based losses. Then, by explicitly controlling magnitude information, we can effectively capture the advantages of popularity information. In the next section, we present analyses that lay the foundation for understanding how magnitude information is encoded.

\section{Understanding the Role of Weight Decay}

In this section, we investigate the role of weight decay by addressing two initial research questions: \textbf{(RQ1)} Is weight decay necessary? \textbf{(RQ2)} What properties does weight decay encode into the embedding matrices? To answer these questions, we conduct a series of empirical studies to assess the influence of weight decay on CF model training. Through this investigation, we reveal weight decay's surprising relationship with popularity. Building on these insights, we perform a theoretical analysis to further explain why this phenomenon occurs.  To begin, we introduce the datasets, models, and losses, used in our empirical analyses, as well as the evaluation strategy. 

\begin{figure*}[t] \centering \begin{subfigure}{0.48\textwidth} \centering \includegraphics[width=0.9\textwidth]{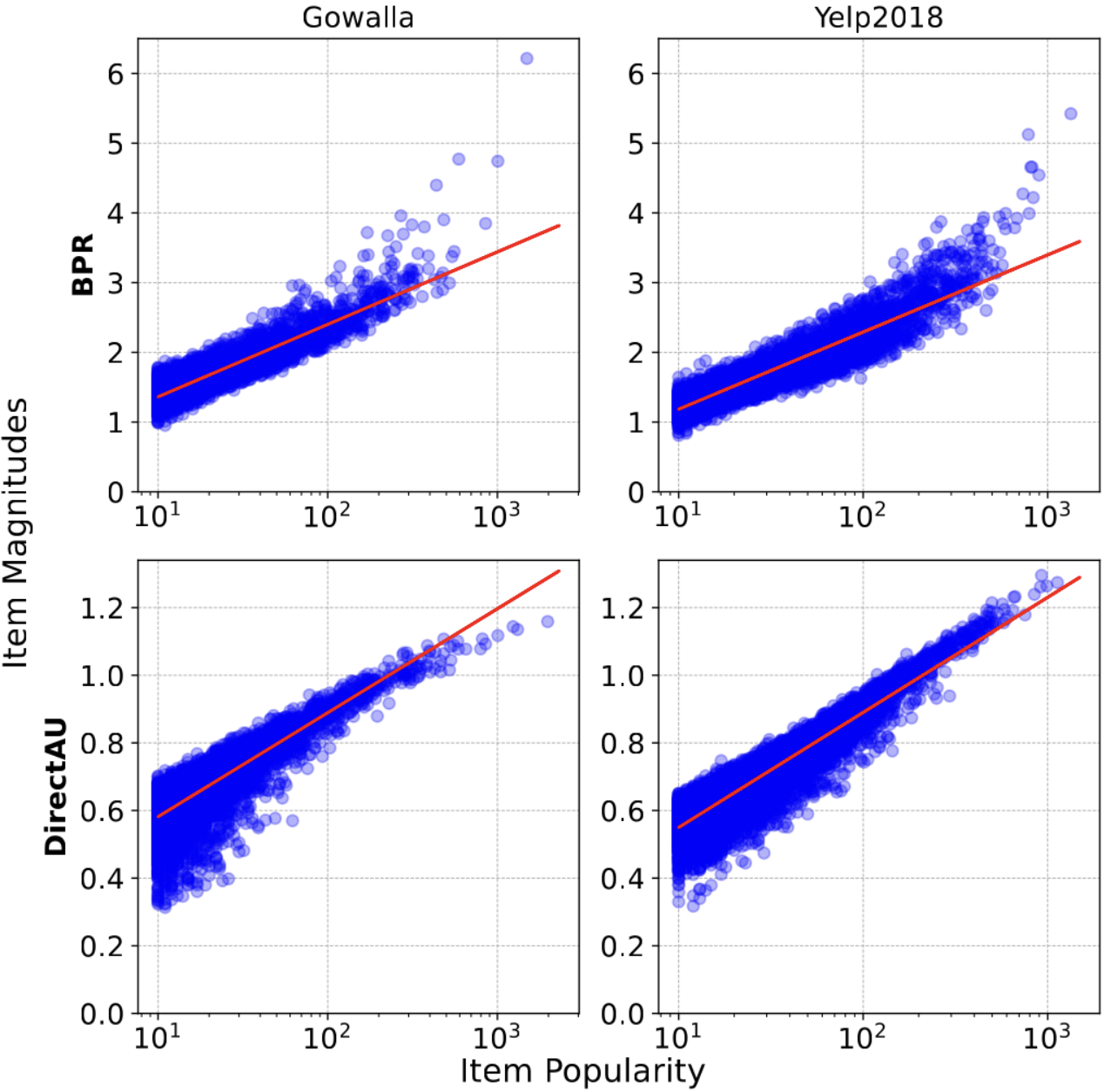} \caption{With Weight Decay ($\lambda = 1 \times 10^{-6}$): The strong correlation indicates that weight decay encodes the popularity information.} \label{fig:motivation_dau_bpr_mag_deg} \end{subfigure} \hspace{0.6cm} \begin{subfigure}{0.48\textwidth} \centering \includegraphics[width=0.9\textwidth]{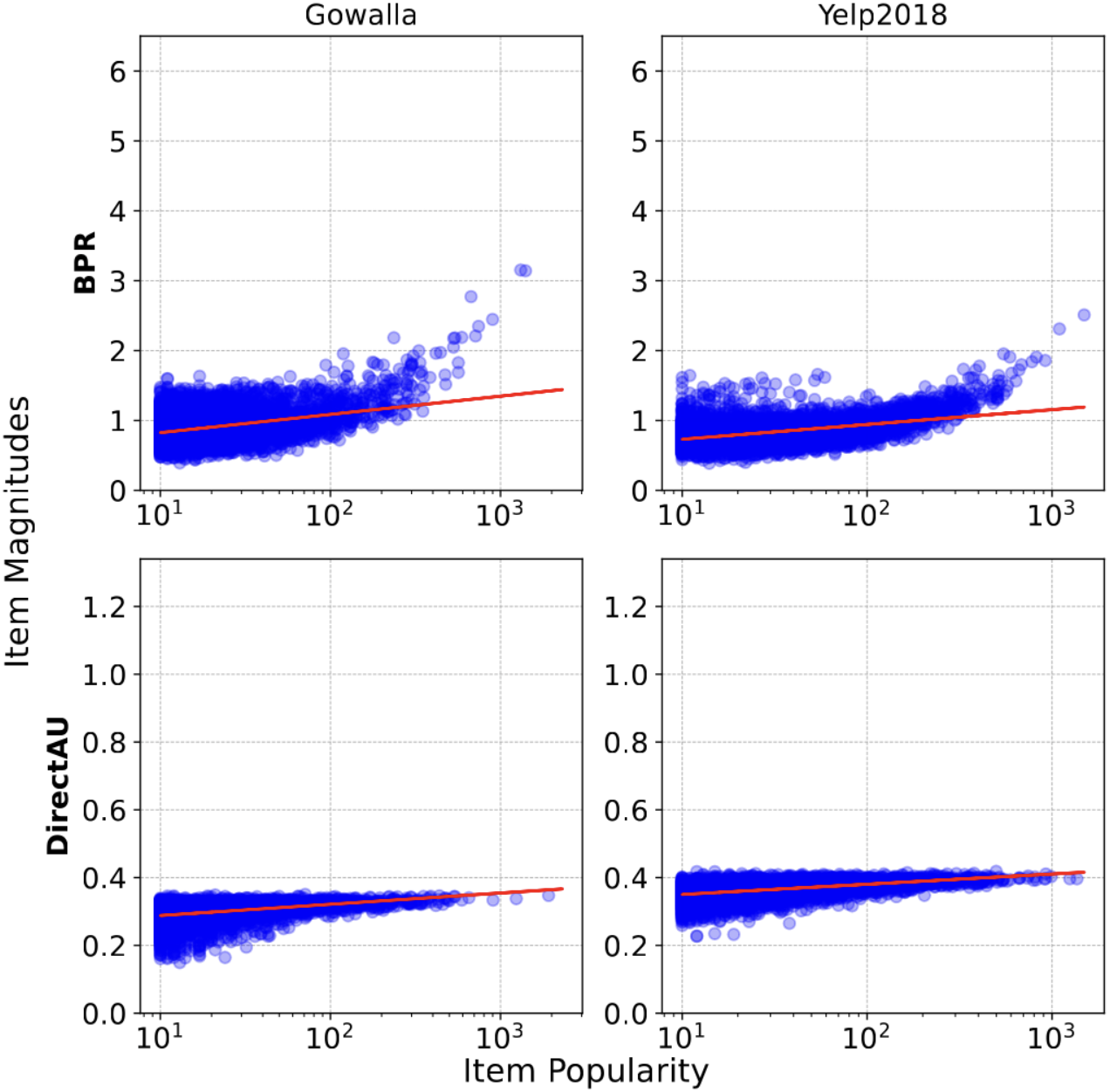} \caption{Without Weight Decay ($\lambda = 0$): The weak correlation highlights that popularity information is not encoded.} \label{fig:motivation_dau_bpr_mag_deg_wd0} \end{subfigure} \caption{Comparison of BPR and DirectAU embedding magnitudes based on the presence or absence of Weight Decay.} \label{fig:combined_motivation} \vspace{-0.3cm} \end{figure*}

\subsection{Empirically Validating  \\ the Importance of Weight Decay}
\label{section:emp_setup}
We start by empirically assessing weight decay, focusing on its influence on performance and answering \textbf{RQ1} and \textbf{RQ2}. 

\vspace{0.1cm}
\noindent \textbf{Experimental Setup.} 
 We focus on four standard recommender system datasets: \texttt{MovieLens1M}\cite{harper2015movielens}, \texttt{Gowalla} \cite{cho2011gowalla}, \texttt{Yelp2018} \cite{yelp}, and \texttt{AmazonBook} \cite{mcauley2016amazonbook}. 
Detailed characterizations of these datasets are available in \Cref{section:app_dataset}. We apply several model backbones to these datasets, including a linear and non-linear MF (utilizes MLPs to transform the embeddings\footnote{Non-linear MF is utilized given the generally poorer results seen in pairwise neural models \cite{rendle2020neuralcollaborativefilteringvs}}). In \Cref{section:app_lightgcn}, we additionally provide results for LightGCN \cite{he2020lightgcn}. The choice of these backbones allows for a wide range of foundational interaction calculations that are seen in many different architectures. We train on four distinct ranking loss functions, $L_{rank}$, including BPR, SSM, DirectAU, and MAWU, and hyper-parameter tuned weight decay. The full loss $L$ is specified as:
\begin{equation}
    L =  \sum_{(u,i) \in \mathcal{B}}  L_{rank}(u, i) +  \frac{\lambda}{2} \left( \sum_u \|\mathbf{u}\|^2 + \sum_i \|\mathbf{i}\|^2 \right),
    \label{eqn:rank_reg_loss}
\end{equation}
where the ranking loss is applied to data batches $B$. The weight decay strength, $\lambda$, is applied to the user and item matrices during each gradient descent step. 
Further implementation details, including architectures, losses, and hyperparameters, can be found in  \Cref{section:app_exp_details}. 

For evaluation, we use the NDCG@20 metric to assess overall performance. Additionally, following \cite{zhang2023logdet}, we conduct a stratified analysis by categorizing items based on their popularity. For each user, the overall NDCG@20 is decomposed into separate NDCG@20 scores for popular, neutral, and unpopular items. Detailed equations and stratification cutoffs are provided in the \Cref{section:app_exp_details}. By assessing how different items contribute to overall performance, we can accurately determine when a model is making use of popularity information, as opposed to the latent features within the embeddings.

\subsubsection{\textbf{(RQ1)} Is Weight Decay Necessary?}

For each of the CF pipelines, we examine scenarios corresponding to different weight decay strengths, denoted as $\lambda$, starting at $0$, and gradually increasing up to $0.0001$. At each level, we calculate overall NDCG@20, along with stratified NDCG@20 values for popular, neutral, and unpopular items. Intuitively, larger discrepancies between these stratified metrics suggest a stronger reliance on popularity information during  ranking.

In \Cref{fig:motivation_dau_bpr_wd}, we present representative results from training MF models with BPR and DirectAU. Note that DirectAU follows the definition of an angle-based loss, while BPR does not. 
Our first observation is a clear dependence on weight decay for DirectAU, where \textbf{a poor choice of weight decay significantly degrades performance}. This finding shows that while DirectAU and other angle-based loss functions can achieve high performance, such outcomes depend on carefully tuning weight decay to effectively encode popularity information, rather than being an intrinsic property of the loss function itself. For instance, when subjected to weak weight decay, DirectAU's performance can degrade to the extent that BPR surpasses it, as observed in datasets like \texttt{MovieLens1M} and \texttt{Gowalla}. Interestingly, these variations in overall performance result from a systematic trade-off of performance on unpopular items for performance on popular items, as reflected in the stratified metrics in \Cref{fig:motivation_dau_bpr_wd}. Moreover, this trade-off is strongly correlated with the strength of weight decay. Thus, our second observation is that \textbf{as weight decay strength increases, models exhibit greater discrepancies in performance between these item categories}. Since BPR can inherently encode popularity information through its dot product similarity function, even without weight decay, we find that the performance gap between popular and unpopular items is less pronounced compared to DirectAU.  
Results for SSM and MAWU, as well as non-linear MF, follow similar trends and are provided in \Cref{section:app_addtl_exp} given space constraints. 

\subsubsection{\textbf{(RQ2)} What properties does weight decay encode? }

To understand the underlying mechanism of this phenomenon, we examine the distribution of magnitude values within the embedding tables. In \Cref{fig:motivation_dau_bpr_mag_deg}, we plot item magnitudes for BPR and DirectAU as a function of item popularity, calculated from the training set, for models trained with moderate weight decay ($\lambda = 1 \times 10^{-6}$). Surprisingly, \textbf{regardless of whether models were trained with BPR or DirectAU, both show a high correlation between the embedding magnitudes and the underlying popularity information}. Since DirectAU's loss is designed to encode meaningful angles between users and items, the process by which this magnitude information is encoded is not immediately clear. To further clarify the role of weight decay, we analyze the magnitude distributions in \Cref{fig:motivation_dau_bpr_mag_deg_wd0} with $\lambda = 0$.  In this scenario, the correlation is significantly reduced, and minimal popularity information is encoded within the embeddings. To formalize why weight decay induces this correlation, we provide a theoretical analysis in the next section showing how popularity information is encoded into the item magnitudes.

\subsection{Theoretical Connection Between Weight Decay and Popularity Information}
\label{section:theory}
Building on our empirical findings demonstrating a relationship between item magnitude and popularity information induced by weight decay, we aim to develop a theoretical explanation for this phenomenon, bolstering our answer to \textbf{RQ1} and \textbf{RQ2}. To achieve this, we analyze the gradient descent process for \Cref{eqn:rank_reg_loss}. We begin by characterizing the probability that an interaction is sampled into a batch, determining its likelihood of receiving a ranking loss update. Using this probability, we then solve for the expected magnitude update under the gradient descent process.

\subsubsection{Probability of Sampling an Interaction in a Batch} As CF models are typically trained using mini-batch gradient descent, only a subset of the full interaction set $E$ is used during each update. However, in many settings, weight decay is applied to the full embedding tables unless a custom implementation is specified. Therefore, we first establish the probability that an interaction is included in the batch and receives a gradient update for both the ranking loss and weight decay. Otherwise, the embedding update is solely from weight decay.

\vspace{0.1cm}
\noindent \textbf{Setup.} We assume that the ranking loss is applied to batches of data $B$, where $|B| \le |E|$. Then, for an item $i$ with $d_{i}$ interactions, i.e. the degree of the item in the interaction graph, \Cref{lemma:prob_batch} specifies the probability that item $i$ will appear in $B$.  
This result is derived in \Cref{section:app_obs_1}. 
\begin{observation} 
    For an item $i$, the probability of appearing in a batch $B$ is given by: 
    \begin{align}
        P(i \in B) = 1 - \left(1 - \frac{|B|}{|E|}\right)^{d_{i}}.
        \label{eqn:prob}
    \end{align}
\label{lemma:prob_batch}
\end{observation}
\noindent \textbf{Takeaway. } \Cref{eqn:prob} demonstrates that an item $i$ is more likely to appear in a batch when either: (a) the item is of high popularity or (b)  the batch size is large relative to the interaction set. In the next section, we use this result to explain how the popularity information is encoded into the embedding vectors.

\subsubsection{Asymmetric Gradient Updates Encode Popularity Information} Leveraging \Cref{lemma:prob_batch}, we now describe the expected gradient descent process for \Cref{eqn:rank_reg_loss}. Notably, when an item is included in a batch, it receives gradient updates from both the ranking loss and weight decay. However, when an item is not in a batch, it only receives an update from weight decay. We refer to this difference in gradients as \textit{asymmetric updates} and demonstrate how this asymmetry encodes popularity information within the embedding vectors.

\vspace{0.1cm}
\noindent \textbf{Setup.} We assume that $L_{rank}$ in \Cref{eqn:rank_reg_loss} is an angle-based loss, specifically focusing on cosine similarity. It is important to note that other losses encoding positive interactions, such as the alignment term in DirectAU, are equivalent to cosine similarity up to a constant when using normalized embeddings \cite{loveland2024understandingscalingcollaborativefiltering}. We then analyze the gradient descent process for \Cref{eqn:rank_reg_loss} to determine the expected change in magnitude for an item embedding $\mathbf{i}$ on the $k^{th}$ gradient descent step.  The expression for this relationship is presented in \Cref{theorem:mag_change}. We use $\eta$ to denote the learning rate and $\lambda$ to denote the strength of weight decay. 
\begin{theorem}
For a user $u$ with interacted item $i$, the expected change in the item's magnitude after the $k^{th}$ gradient descent step, $\mathbb{E}[\Delta\|\mathbf{i}^{(k+1)}\|^{2}]$, when trained with asymmetric updates of cosine similarity and weight decay, is expressed as: 
\begin{align} 
 \mathbb{E}[\Delta\|\mathbf{i}^{(k+1)}\|^{2}] &=  \underbrace{\mathbb{E}[\|\mathbf{i}^{(k)}\|^2] \eta\lambda(\eta\lambda - 2)}_{\text{Weight Decay Contribution}} +  \\   
 &\quad  \underbrace{\left(1 - \left(1 - \frac{|B|}{|E|}\right)^{d_{i}} \right)}_{\text{P(i $\in$ B)}}
 \underbrace{\mathbb{E} \Bigl[ \frac{\eta^{2}}{\|\mathbf{i}^{(k)}\|^{2}} \left(1 - \frac{(\mathbf{u}^{(k)} \cdot \mathbf{i}^{(k)})^2}{\|\mathbf{u}^{(k)}\|^2 \|\mathbf{i}^{(k)}\|^2}\right) \Bigr]. }_{\text{Ranking Loss Contribution}} \notag
\end{align}
\label{theorem:mag_change}
\end{theorem}
\begin{proof}
We give the proof for this result in~\Cref{section:theorem1_proof}.
\end{proof}
\noindent \textbf{Takeaway. } 
The first term reflects the change in magnitude resulting from weight decay, while the second term represents the change in magnitude due to cosine similarity. Although weight decay lowers an item's magnitude during each gradient descent iteration, we show \Cref{section:theorem1_proof} that an item's magnitude increases as a result of the cosine similarity (This is also seen in the SSL setting \cite{draganov2024hiddenpitfallscosinesimilarity} ). Moreover, the ratio of increases provided by cosine similarity, relative to the decreases from weight decay, is governed by the probability established in \Cref{lemma:prob_batch}. Thus, one can expect that an \textbf{item $i$'s popularity level, $d_{i}$, will be encoded into $i$'s embedding via its magnitude}. In \Cref{section:theorem1_proof}, we further analyze this equation for different levels of $|B|$ and $d_{i}$ to demonstrate how typical dataset properties and batch sizes contribute to strong popularity encoding. Note that if weight decay is not applied, both terms are scaled by $d_{i}$, mitigating the asymmetric updates that encode popularity encoding (shown in \Cref{corollary1}). Additionally, in \Cref{corollary2} we provide a corollary of \Cref{theorem:mag_change} showing that negative sampling can partially alleviate popularity encoding, but requires the full set of non-interacted items to be used as negative samples to fully solve the issue.

\section{Removing the Need for Weight Decay}

In the previous section, we established that weight decay's core function is to encode popularity information. With this understanding, we first explore potential alternatives to weight decay that leverage existing strategies, such as batching or message passing. Although we show that weight decay can be removed in certain settings with such strategies, the learning setups typically require significant concessions on performance or speed. 
%
%
Thus, we then tackle the following question:
\textbf{(RQ3)} Can the information typically learned through weight decay be directly encoded into the embedding matrices \textit{without} a performance or speed compromise? Our goal is to remove the need for weight decay in the training process, thereby simplifying optimization and eliminating the necessity for extensive hyperparameter tuning to find the optimal weight decay strength. Leveraging our new insights, we introduce PRISM, a strategy that pre-encodes popularity information directly, effectively decoupling magnitude and angle learning.

\subsection{Previous Alternatives to Weight Decay}

We first explore existing methods to simplify or replace weight decay while highlighting their limitations. We begin by examining a batched weight decay that does not require full embedding updates and then consider certain losses and architectures that explicitly encode magnitude information.

\subsubsection{Batched Weight Decay}
In large-scale settings, one strategy is to apply weight decay at the batch-level, meaning weight decay only updates users or items in the current batch. This approach treats all items equally (setting $P(i \in B)$ to 1 in \Cref{theorem:mag_change}), removing the dependency on popularity for magnitude changes. However, as shown in \Cref{fig:batched_wd}, performance with batched weight decay is extremely poor, matching the performance when $\lambda = 0$. Thus, we conclude that the \textbf{batched weight decay is not an effective strategy to reduce the computational demands of full weight decay}. This result also highlights that the benefit of weight decay does not inherently arise from the regularization itself, but rather from the asymmetric updates that encode popularity. This conclusion is further supported by the magnitude distributions in \Cref{fig:batched_wd} which closely resemble those in \Cref{fig:motivation_dau_bpr_mag_deg_wd0} when $\lambda = 0$.

\subsubsection{Strategies that Directly Encode Magnitude}

We have shown that weight decay's benefit lies in encoding popularity into the magnitudes. Thus, we hypothesize that methods that directly learn magnitudes do not gain as much benefit from weight decay. From the perspective of losses, those involving the dot product can inherently learn magnitude information, as seen in \Cref{fig:motivation_dau_bpr_wd}, where BPR's overall performance changes are small compared to DirectAU. However, \textbf{BPR also generally underperforms, making it a poor alternative to weight decay}. Additionally, certain architectures, such as LightGCN, can encode magnitude during learning, regardless of loss~\cite{wu2022ssm}. We verify our hypothesis that such architectures would not benefit from weight decay by examining LightGCN’s performance across different weight decay levels. Our analysis (\Cref{fig:lightgcn_app}) shows LightGCN maintains relatively consistent performance across  weight decay levels, indicating its independence from weight decay. However, \textbf{although LightGCN does not require weight decay, it is also challenging to integrate into large-scale pipeline}, making it an impractical solution for learning magnitude information.

\begin{figure}[t]
    \centering
    \includegraphics[width=8.cm]{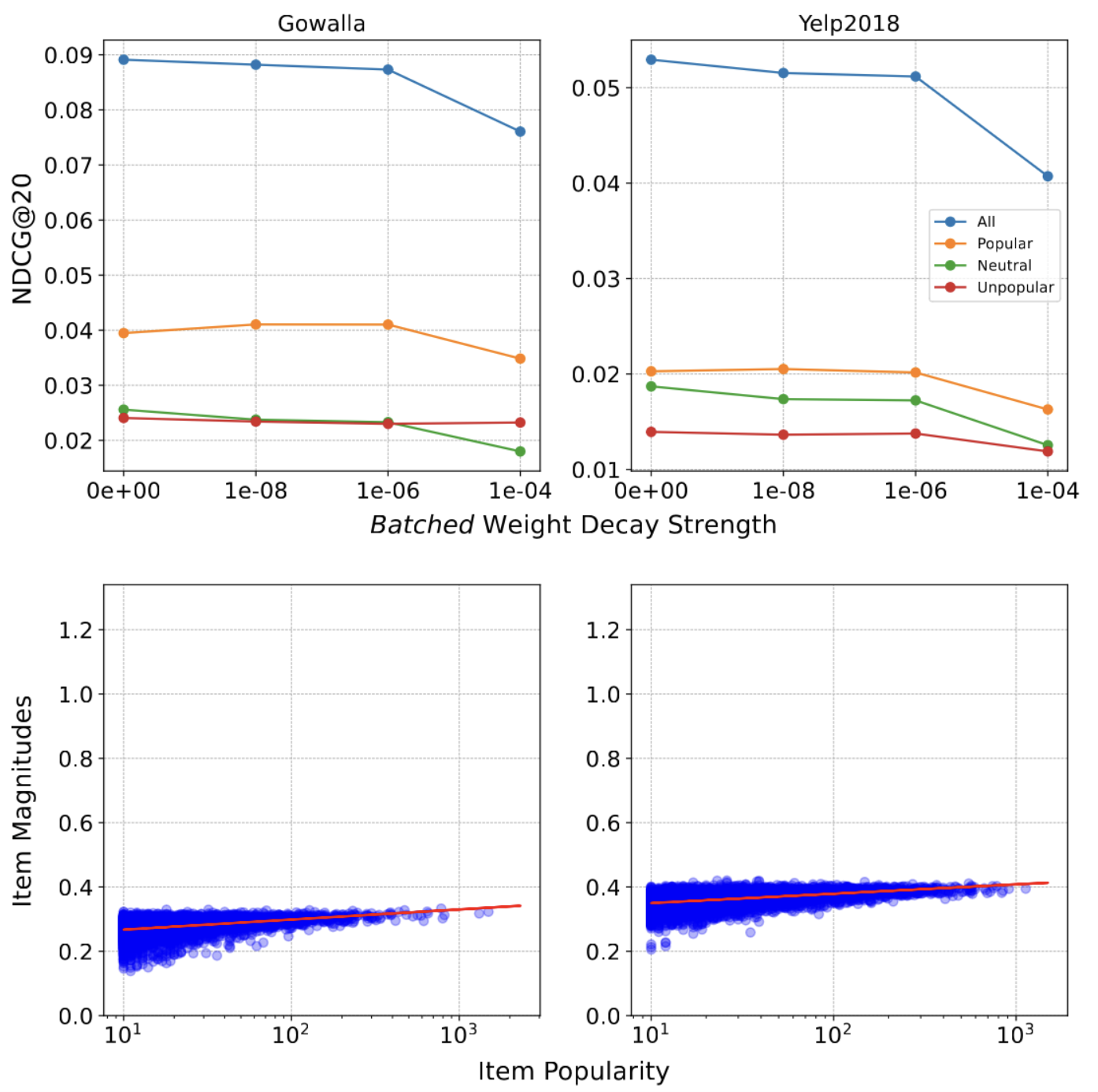}
    \caption{\textbf{DirectAU (angle-based loss) Performance with Varying \textit{Batched} Weight Decay}.
    The performance across batched weight decay strengths are highly similar and match performance without weight decay ($\lambda = 0$). Item magnitudes for Batched Weight Decay strength equal to $1 \times 10^{-6}$ is provided for comparison.}
    \label{fig:batched_wd}
\end{figure}

\begin{table*}[t]
\centering
\small
\caption{Performance comparison between models trained with weight decay and PRISM. Highly similar performances demonstrate that PRISM is effective at replacing weight decay without the need for any tuning of parameter updates.}
\label{table:wd_vs_init}
\arrayrulecolor{black} 
\begin{tabular}{@{} llccccc @{}} 
\toprule 
\textbf{Backbone} & \textbf{Loss} & \textbf{Training Method} & \textbf{MovieLens1M} & \textbf{Gowalla} & \textbf{Yelp2018} & \textbf{AmazonBook} \\ 
\midrule 
\multirow{9}{*}{\textbf{MF}} 

& \multirow{3}{*}{\textbf{SSM}}  
 & Tuned Weight Decay & 0.3061 $\pm$ 0.0004 & 0.0999 $\pm$ 0.0068 & 0.0616 $\pm$ 0.0070 & 0.0609 $\pm$ 0.0066 \\  
 & & PRISM & 0.3002 $\pm$ 0.0018 & 0.1033 $\pm$ 0.0021 & 0.0605 $\pm$ 0.0002 & 0.0598 $\pm$ 0.0049 \\  
\arrayrulecolor{lightgray}\cmidrule[\lightrulewidth]{3-7} 
\arrayrulecolor{black}
 & & \cellcolor{gray!20}\% Diff & \cellcolor{gray!20}-1.97\% & \cellcolor{gray!20}3.30\% & \cellcolor{gray!20}-1.82\% & \cellcolor{gray!20}-1.84\% \\
\cmidrule{2-7} 
 & \multirow{3}{*}{\textbf{DirectAU}} 
 & Tuned Weight Decay & 0.2406 $\pm$ 0.0000 & 0.1176 $\pm$ 0.0003 & 0.0735 $\pm$ 0.0002 & 0.0719 $\pm$ 0.0003 \\  
 & & PRISM  & 0.2304 $\pm$ 0.0011 & 0.1210 $\pm$ 0.0000 & 0.0745 $\pm$ 0.0002 & 0.0755 $\pm$ 0.0004 \\  
\arrayrulecolor{lightgray}\cmidrule[\lightrulewidth]{3-7} 
\arrayrulecolor{black}
 & & \cellcolor{gray!20}\% Diff & \cellcolor{gray!20}-4.42\% & \cellcolor{gray!20}2.81\% & \cellcolor{gray!20}1.34\% & \cellcolor{gray!20}4.77\% \\
\cmidrule{2-7} 
 & \multirow{3}{*}{\textbf{MAWU}} 
 & Tuned Weight Decay & 0.2455 $\pm$ 0.0011 & 0.1332 $\pm$ 0.0006 & 0.0745 $\pm$ 0.0015 & 0.0877 $\pm$ 0.0002 \\  
 & & PRISM & 0.2376 $\pm$ 0.0005 & 0.1331 $\pm$ 0.0014 & 0.0739 $\pm$ 0.0016 & 0.0891 $\pm$ 0.0004 \\  
\arrayrulecolor{lightgray}\cmidrule[\lightrulewidth]{3-7} 
\arrayrulecolor{black}
 & & \cellcolor{gray!20}\% Diff & \cellcolor{gray!20}-3.33\% & \cellcolor{gray!20}-0.08\% & \cellcolor{gray!20}-0.81\% & \cellcolor{gray!20}1.57\% \\
\midrule 
\multirow{9}{*}{\textbf{Non-Linear MF}} 

& \multirow{3}{*}{\textbf{SSM}}  
 & Tuned Weight Decay & 0.2785 $\pm$ 0.0013 & 0.0860 $\pm$ 0.0010 & 0.0502 $\pm$ 0.0002 & 0.0343 $\pm$ 0.0005 \\  
 & & PRISM & 0.2330 $\pm$ 0.0025 & 0.0728 $\pm$ 0.0006 & 0.0389 $\pm$ 0.0023 & 0.0287 $\pm$ 0.0001 \\  
\arrayrulecolor{lightgray}\cmidrule[\lightrulewidth]{3-7} 
\arrayrulecolor{black}
 & & \cellcolor{gray!20}\% Diff & \cellcolor{gray!20}-19.53\% & \cellcolor{gray!20}-18.10\% & \cellcolor{gray!20}-29.04\% & \cellcolor{gray!20}-19.51\% \\
\cmidrule{2-7} 
 & \multirow{3}{*}{\textbf{DirectAU}} 
 & Tuned Weight Decay & 0.2318 $\pm$ 0.0006 & 0.0989 $\pm$ 0.0012 & 0.0646 $\pm$ 0.0004 & 0.0573 $\pm$ 0.0004 \\  
 & & PRISM & 0.2502 $\pm$ 0.0012 & 0.1105 $\pm$ 0.0024 & 0.0698 $\pm$ 0.0005 & 0.0714 $\pm$ 0.0004 \\  
\arrayrulecolor{lightgray}\cmidrule[\lightrulewidth]{3-7} 
\arrayrulecolor{black}
 & & \cellcolor{gray!20}\% Diff & \cellcolor{gray!20}7.36\% & \cellcolor{gray!20}10.51\% & \cellcolor{gray!20}7.44\% & \cellcolor{gray!20}19.77\% \\
\cmidrule{2-7} 
 & \multirow{3}{*}{\textbf{MAWU}}  
 & Tuned Weight Decay & 0.2405 $\pm$ 0.0029 & 0.1274 $\pm$ 0.0029 & 0.0670 $\pm$ 0.0001 & 0.0800 $\pm$ 0.0006 \\  
 & & PRISM & 0.2600 $\pm$ 0.0006 & 0.1326 $\pm$ 0.0015 & 0.0691 $\pm$ 0.0003 & 0.0881 $\pm$ 0.0001 \\  
\arrayrulecolor{lightgray}\cmidrule[\lightrulewidth]{3-7} 
\arrayrulecolor{black}
 & & \cellcolor{gray!20}\% Diff & \cellcolor{gray!20}7.50\% & \cellcolor{gray!20}3.91\% & \cellcolor{gray!20}3.04\% & \cellcolor{gray!20}9.19\% \\
\bottomrule 
\end{tabular} 
\end{table*}

\subsection{PRISM: \underline{P}opularity-awa\underline{R}e \underline{I}nitialization \underline{S}trategy for Embedding \underline{M}agnitudes }

To eliminate the need for weight decay, we propose \textbf{PRISM}, a strategy that pre-encodes learned magnitude information during initialization. By incorporating this information upfront, the optimization process can focus on learning angle information, rather than encoding popularity. We note that magnitude modulation has traditionally been an in-processing strategy, often employed to reduce popularity bias  \cite{Klimashevskaia_2024}. Therefore, our pre-processing approach opens a new avenue for improving training, facilitated by insights from our analysis. PRISM specifically leverages the observed log-linear relationship between popularity and magnitudes (as seen in \Cref{fig:motivation_dau_bpr_mag_deg} and \Cref{theorem:mag_change}), and initializes the embedding magnitudes using the log of their corresponding popularity. Specifically, we start with initial embedding vectors $\mathbf{i}_{init}$, which can be generated through any initialization strategy, and normalizes them to unit vectors. PRISM then scales these unit vectors according to their popularity. Thus, for an item $i$'s embedding $\mathbf{i}$ with popularity level $d_{i}$, the initialization is give by: 

\begin{equation}
    \mathbf{i} = \left(\alpha \cdot \log(d_{i} + 2) + (1 - \alpha)\right) \cdot \frac{\mathbf{i}_{init}}{\|\mathbf{i}_{init}\|},
    \label{eqn:mag_init}
\end{equation}
where the trade-off parameter $\alpha$ modulates the strength of the encoding, allowing for a linear interpolation between fully encoding popularity information ($\alpha=1$) and simply initializing as unit vectors ($\alpha=0$). We use $\log(d_{i} + 2)$ to ensure that cold-start items are not reduced to a zero-magnitude vector. Next, we demonstrate that PRISM is able to achieve the performance of models tuned for weight decay by simply setting $\alpha = 1$ and initializing the embedding magnitudes. 
\section{Empirical Analysis of PRISM}

In this section, we leverage PRISM and establish how effective it is at replacing weight decay. We additionally provide analysis on how PRISM can be used to expedite training (\S~\ref{sec:speed}), as well as improve the performance of unpopular items, mitigating popularity bias (\S~\ref{sec:bias}) \footnote{Code can be found at \href{https://github.com/snap-research/PRISM/}{https://github.com/snap-research/PRISM/}}. 

\subsection{Effectiveness of Replacing Weight Decay}

Employing a similar setup as described in \Cref{section:emp_setup}, we eliminate weight decay from training and instead initialize embedding vectors according to \Cref{eqn:mag_init}, with $\alpha = 1$.  Notably, this approach eliminates the need for: (a) hyperparameter tuning beyond model- or loss-specific parameters, and (b) full access to the user and item embedding tables. As in \Cref{fig:motivation_dau_bpr_wd}, we compare both overall and stratified performance metrics for models using our initialization strategy against those utilizing traditional tuned weight decay.

In \Cref{table:wd_vs_init}, we compare the performance of models tuned with weight decay against those trained with PRISM. We emphasize that with PRISM, a single model is trained without the need for tuning. In both linear and non-linear MF models, \textbf{PRISM generally demonstrates comparable or even superior performance to models trained with weight decay}. 
For example, linear MF models using PRISM show an average performance change of just 0.12\%, indicating nearly identical aggregate performance. Additionally, for non-linear MF models trained with DirectAU and MAWU, PRISM delivers an 8.59\% performance boost. These findings underscore PRISM's practical benefits, demonstrating that we can eliminate weight decay without sacrificing performance, particularly for SOTA losses. However, we note that non-linear MF models trained with SSM display an average performance decline of 21.61\%. This likely stems from SSM's heavy reliance on popularity information \cite{wu2022ssm}, which becomes challenging to recover after passing through MLP layers. This is further explained in \Cref{corollary3} where we show how SSM lacks a mechanism to maintain popularity information through the MLPs.   Despite this, we anticipate that PRISM will prove advantageous in most practical applications, and we attribute this specific issue more to the combination of SSM and non-linear MLPs.

\subsection{Benefits Beyond Performance}

In this section, we explore the advantages of removing weight decay from the training process, beyond merely matching the performance of weight-decay-tuned models. Specifically, we pose the question: \textbf{(RQ4)} What opportunities does our initialization strategy present? We focus on two main areas: training efficiency and the mitigation of popularity bias. For this analysis, we use DirectAU as a representative angle-based loss to illustrate these additional benefits.

\begin{figure*}[t!] 
    \centering \includegraphics[width=17.7cm]{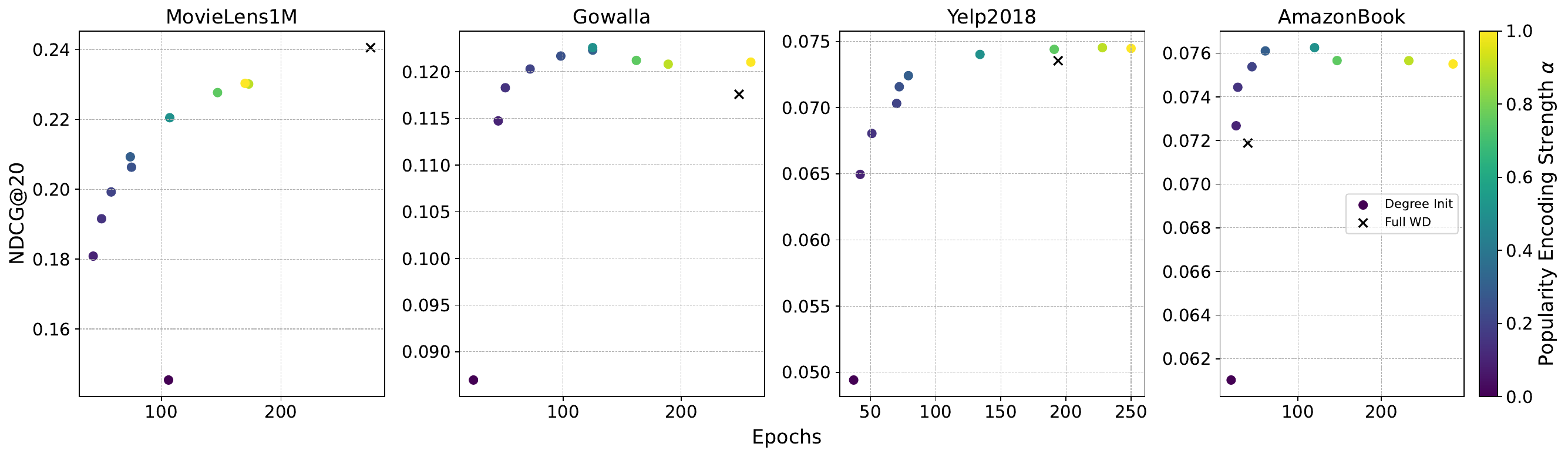}
    \vspace{-0.3cm}
    \caption{\textbf{DirectAU Performance versus Number of Training Epochs to Early Stopping}. Tuning $\alpha$ allows for models with comparable performance to train significantly faster. Moreover, higher $\alpha$ tends to create higher performing models beyond what is attainable with weight decay. } 
\label{fig:perf_vs_epoch} 
\end{figure*}

\begin{figure*}[t!] 
    \centering \includegraphics[width=17.7cm]{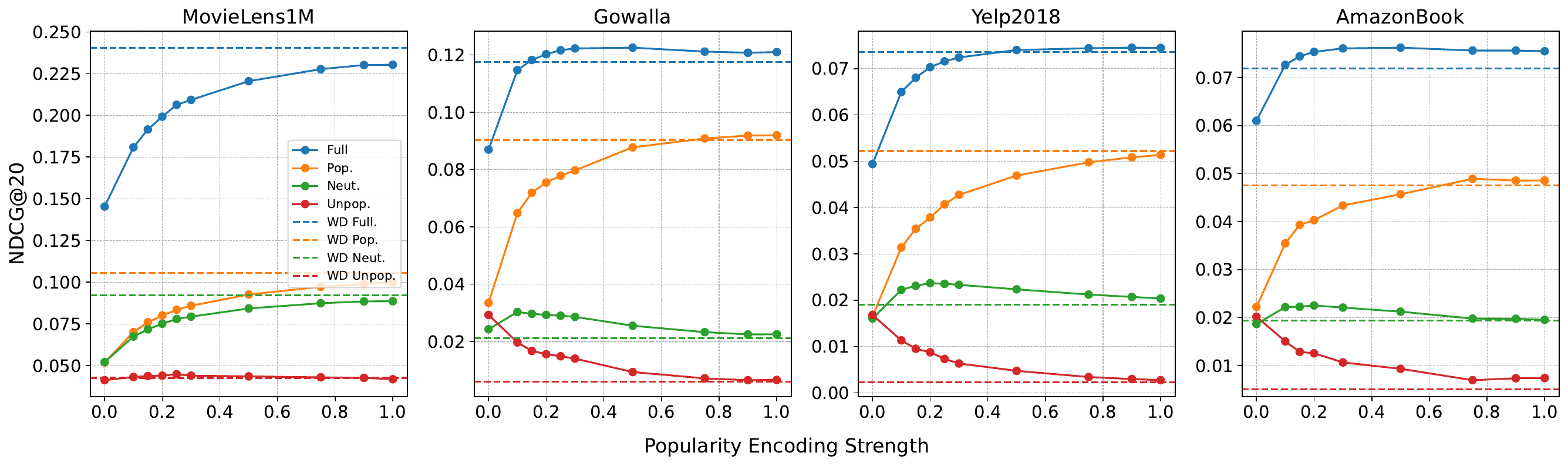} 
    \vspace{-0.3cm}
    \caption{\textbf{DirectAU Performance with Varying $\alpha$}. Decreasing $\alpha$ naturally trades-off popular for unpopular item performance with minimal loss to overall performance.} 
\label{fig:alpha_tune} 
\end{figure*}

\subsubsection{Training Speed Benefits}
\label{sec:speed}
By appropriately tuning $\alpha$ in PRISM, we demonstrate that models can train significantly faster, as this eliminates extended training periods typically devoted to encoding popularity. In \Cref{fig:perf_vs_epoch} we present the performance relative to the number of training epochs for various popularity encoding strengths $\alpha$. Notably, several $\alpha$ levels not only accelerate training but also outperform models fully trained with weight decay. While the underlying reasons for this behavior will be discussed in the next section, we emphasize that across all datasets, models trained with PRISM achieve comparable performance with a significantly faster training process with appropriate $\alpha$. For instance, Gowalla exhibits the largest reduction in training time, with a 49.80\% improvement going from 249 to 125 training epochs. \textbf{On average, PRISM achieves training speed improvements of 38.48\% for comparable performing models}. Although some datasets, such as AmazonBook, show an increase in training epochs with higher $\alpha$, this tends to additionally enhance performance. We attribute this to the model continuing to learn more semantically meaningful angles when popularity information has already been encoded. This is shown in \Cref{fig:training_curves}, in the Appendix, where the cosine similarity measure quickly saturates early during training, while models spend considerable epochs gradually encoding popularity information, indicated by a continuous rise in dot product metrics. Thus, through PRISM, models are able to focus on learning meaningful angles for a longer duration, enhancing performance.

\subsubsection{Popularity Bias Mitigation Benefits}
\label{sec:bias}

Second, we investigate the impact on popularity bias. We find that by directly adjusting $\alpha$, we can effectively mitigate this bias and provide a meaningful mechanism to trade off preference toward popular or unpopular items. 
We demonstrate this phenomenon in \Cref{fig:alpha_tune} where we observe that PRISM can produce a variety of high-quality models that maintain similar overall performance by effectively balancing the trade-off between popular and unpopular item performance. Although one might expect that reducing the performance of popular items would result in a net decline, our findings indicate that the improvements in neutral and unpopular item performance can lead to an overall positive effect. To evaluate the extent of improvement in unpopular item performance, we analyze the ratio of unpopular NDCG@20 to popular NDCG@20, which measures the relative contribution of each to the overall performance. On average, PRISM achieves an average 75.98\% improvement in this metric for the best performing model (selected by validation NDCG@20) alongside an additional average increase of 1.7\% in overall performance for these de-biased models. Thus, our results show that \textbf{$\alpha$ serves as an effective knob, allowing users to prioritize either popular or unpopular items without compromising performance}. We further compare PRISM to three de-biasing methods in \Cref{section:app_debias}, showing comparable improvements over these baselines.

\section{Conclusion}

In this work, we uncovered a surprising link between weight decay and the encoding of popularity information into embedding magnitudes. Our analysis shows that weight decay is crucial for achieving state-of-the-art (SOTA) performance, especially in losses lacking alternative mechanisms for encoding popularity. Given the performance sensitivity to weight decay and the costs associated with tuning, we propose PRISM, a Popularity-aware Initialization Strategy for embedding Magnitudes. PRISM directly integrates popularity information into embedding magnitudes, effectively replacing the need for weight decay and its associated hyperparameter tuning. By decoupling angle and magnitude learning, PRISM enables models to achieve comparable or superior performance with faster convergence. Additionally, by tuning the popularity encoding strength of PRISM, we demonstrate the potential for mitigating popularity bias and enhancing the ranking of less popular items. Our work not only revises traditional CF training practices but also provides new insights into the underlying reasons for certain behaviors. Furthermore, our practical solution through PRISM facilitates lightweight, scalable, and efficient training of recommendation systems.

\begin{acks}
This material in this work is supported by the National Science Foundation under a Graduate fellowship, IIS~2212143 and  CAREER Grant No.~IIS~1845491. Any opinions, findings, and conclusions or recommendations expressed in this material are those of the author(s) and do not necessarily reflect the views of the National Science Foundation or other funding parties.
\end{acks}

\bibliographystyle{ACM-Reference-Format}
\bibliography{references}

\newpage
\appendix

\clearpage

\section{Appendix}

\subsection{Dataset Statistics and Information}
\label{section:app_dataset}

Below we provide information regarding the datasets used within the paper. Specifically, we provide their statistics in \cref{tab:dataset_statistics}, focusing on the number of users and items, and number of interactions. We additionally provide the density of each dataset as \# Interactions/(\# Users $\times$ \# Items), which is discussed within the empirical analysis.

\begin{table}[h!]
\centering
\caption{Dataset Statistics.}
\begin{tabular}{lcccc}
\toprule
Dataset & \# Users & \# Items & \# Interactions & Density \\
\midrule
MovieLens1M & 6,040 & 3,629 & 836,478 & 3.82\% \\
Gowalla & 29,858 & 40,981 & 1,027,370 & 0.08\% \\
Yelp2018 & 31,668 & 38,048 & 1,561,406 & 0.13\% \\
AmazonBook & 52,643 & 91,599  & 2,984,108 & 0.06\% \\
\bottomrule
\end{tabular}
\label{tab:dataset_statistics}
\end{table}
For the stratified analysis, we follow \cite{zhang2023logdet} and group items into popular, neutral, and unpopular as the top 5\%, the top 5-20\% and bottom 80\%, respectively. The stratified NDCG@K calculation in provided in \Cref{section:app_eval}.

\section{Experimental Details}
\label{section:app_exp_details}
In this section we provide details on the experimental setup. While this is primarily introduced in \Cref{section:emp_setup}, we use this setup throughout the paper unless otherwise specified. We begin with the hyperparameters considered in the work, then go on to discuss evaluation and implementations. 

\subsection{Hyperparameters and Tuning} For all models, the embedding tables are initialized using PyTorch's uniform Xavier strategy. For experiments where we utilize PRISM, the embeddings are normalized row-wise, and then scaled via the popularity information. Cross-validation over NDCG@K performance is used to find the best model, searching over learning rates $\{0.1, 0.01, 0.001\}$. and weight decays $\{0.0, 1e^{-4}, 1e^{-6}, 1e^{-8}\}$. The embedding dimensions are kept at 64. While linear MF has no additional parameters, we include two MLP layers for non-linear MF with a ReLU activation. For LightGCN, we set a depth of 3. The models are trained for up to 1000 epochs, with early stopping employed over validation NDCG@K (setting K=20) with a patience of 10. For experiments which utilize SSM, we set a negative sampling ratio of 10. For DirectAU, we additionally cross-validate $\gamma$ values from $\{1.0, 2.0, 5.0\}$, as recommender in the original paper \cite{wang2022towards}. For MAWU, we utilize the recommended hyper-parameters for the datasets provided, and only cross-validate the loss-specific hyper-parameters for MovieLens in the range specified by the paper. The batch size for BPR, SSM, DirectAU, and MAWU training are set to roughly maximize size that can fit within memory, which is 16384 for BPR and SSM, and 4096 for DirectAU and MAWU. The train/val/test splits are random and use 80\%/10\%/10\% of the data. 

\subsection{Evaluation} 
\label{section:app_eval}
To evaluate our models, we look at overall NDCG@K, as well as stratified NDCG@K. NDCG@K is calculated for a particular user $u$ and interaction matrix $\mathbf{E}$ as: 

\begin{equation}
   \text{NDCG@K}_{u} = \frac{\text{DCG@K}_{u}}{\text{IDCG@K}_{u}}
\end{equation}
and 
\[
   \text{DCG@K}_{u} = \sum_{i=1}^{K} \frac{\mathbf{E}_{u, i}}{\log_2(i+1)}.
\]
To attain $\text{IDCG@K}_{u}$, we compute the optimal $\text{DCG@K}_{u}$ for the $K$ retrieved items by sorting them in descending order relative to their relevancy. We note that we let $K = \text{min}(20, N(u))$, where $N(u)$ is the number of elements a user $u$ interacts with. Then, each user's NDCG value can span the full range from 0 to 1. The final metric is the average over all of the users. 

For our stratified metrics, we decompose the full NDCG@K metric based on the different popularity sets of items. Specifically, we retain the same IDCG@K across the stratified metrics, but mask out values in $\mathbf{E}_{u,i}$ (setting them to $0$) if the item is not in the popularity group that the metric is computed over. Thus, when summing each of the stratified NDCG@K metrics for a particular user, the resulting sum is equal to the full NDCG@K. 

\subsection{Implementation} The MF backbones and losses are implemented within vanilla PyTorch. LightGCN is implemented using PyTorch Geometric. Data loading and batching is additionally implemented with PyTorch Geometric's dataloader. We use approximate negative sampling for BPR and SSM, as seen in PyG's documentation for their \href{https://pytorch-geometric.readthedocs.io/en/latest/modules/loader.html#torch_geometric.loader.LinkNeighborLoader}{LinkNeighborLoader}, meaning there is a small chance some negative samples may be false negatives. Each model is trained on a single NVIDIA Volta V100 GPU. 

\subsection{Loss Function Definitions}
\label{section:app_loss_funcs}

The loss functions used in the paper are outlined below.

\subsubsection{Bayesian Personalized Rank (BPR)} BPR is one of the earliest recommendation loss functions, operating by maximizing the difference between interacted and non-interacted items for each user \cite{rendle2009bpr, ding2018bprimprove}. Specifically, for a user $u$, interacted item $i$, and non-interacted item $i'$, the BPR loss is:
\vspace{0.1cm}
\begin{equation} 
    L_{BPR} =  \sum_{(u,i,i') \in \mathcal{D}} \ln \sigma(\hat{e}_{u,i} - \hat{e}_{u,i'}), 
\end{equation}
where $\hat{e}_{u, i}$ is the similarity between $u$ and $i$. BPR's original implementation uses the dot product between $u$ and $i$ to attain similarity scores. 

\subsubsection{Sampled Softmax (SSM)} SSM is a set-wise loss which generalizes the single negative sample seen in BPR \cite{wu2022ssm, oord2018representation}. SSM leverages a subset of negative samples, increasing the expressive power and regularization strength. SSM is expressed as: 
\vspace{0.1cm}
\begin{equation} 
    L_{SSM} =  \sum_{(u, i) \in \mathcal{E}} -\log \frac{\text{exp}(\hat{e}_{u,i})}{\text{exp}(\hat{e}_{u,i}) + \sum_{i' \in S} \text{exp}(\hat{e}_{u,i'})}, 
\end{equation}
where $S$ is the set of negative samples. Typically, SSM is equipped with cosine similarity to compute $\hat{e}_{u,i}$.

\subsubsection{DirectAU} DirectAU utilizes an alignment and uniformity term to both encourage positive pairs to have similar embeddings, while scattering uninteracted pairs on the hypersphere. The alignment component of DirectAU is specified as:
\vspace{0.1cm}
\begin{equation}
    L_{align} = \sum_{(u,i)\in\mathcal{E}} \|\mathbf{u} - \mathbf{i}\|^2,
\end{equation}
where $(u,i)$ are observed user-item interactions. The uniformity term is given by
\vspace{0.1cm}
\begin{equation}
    L_{uniform} = \log \sum_{(u,u')\in U} e^{-2\|\mathbf{u}- \mathbf{u'}\|^2} + 
                  \log \sum_{(i,i')\in I} e^{-2\|\mathbf{i}- \mathbf{i'}\|^2}.
\end{equation}
We point out that the embeddings are normalized for the alignment and uniformity term, meaning that DirectAU is an angle-based loss despite using metrics that typically depend on magnitude. 

\subsubsection{MAWU} MAWU extends DirectAU by including margin parameters on each user and item to increase discriminative power over interacted and non-interacted pairs. Typically, this is applied on top of cosine similarity. Moreover, the user and item uniformity terms are each given unique hyper-parameters. Thus, the core difference is in the alignment term, expressed as:
\vspace{0.1cm}
\begin{equation}
    L_{margin-align} = -\sum_{(u,i)\in\mathcal{E}} \text{cos}(\theta_{u, i} + M_{u} + M_{i}),
\end{equation}
where $M_{u}$ and $M_{i}$ denote the learned margin variables. 

\subsection{Validation Performance Plots}

To understand when weight decay provides value during training, we track the validation NDCG@20 as calculated with both the cosine similarity and dot product metrics. Intuitively, performance under cosine similarity takes advantage solely of the learned angle information, while the dot product metric allows for both the magnitude and angle information. In \Cref{fig:training_curves} we give an example set of training curves for DirectAU trained on Gowalla and Yelp2018. We can see that for the best performing weight decay levels, a significant portion of time is spent improving the dot product-based NDCG@20, while the cosine similarity-based NDCG@20 saturates relatively quickly. From this, we argue that a significant portion of training time tends to come from weight decay encoding popularity information, as opposed to learning meaningful angles from DirectAU. 

\begin{figure}[]
    \centering
    \includegraphics[width=8.4cm]{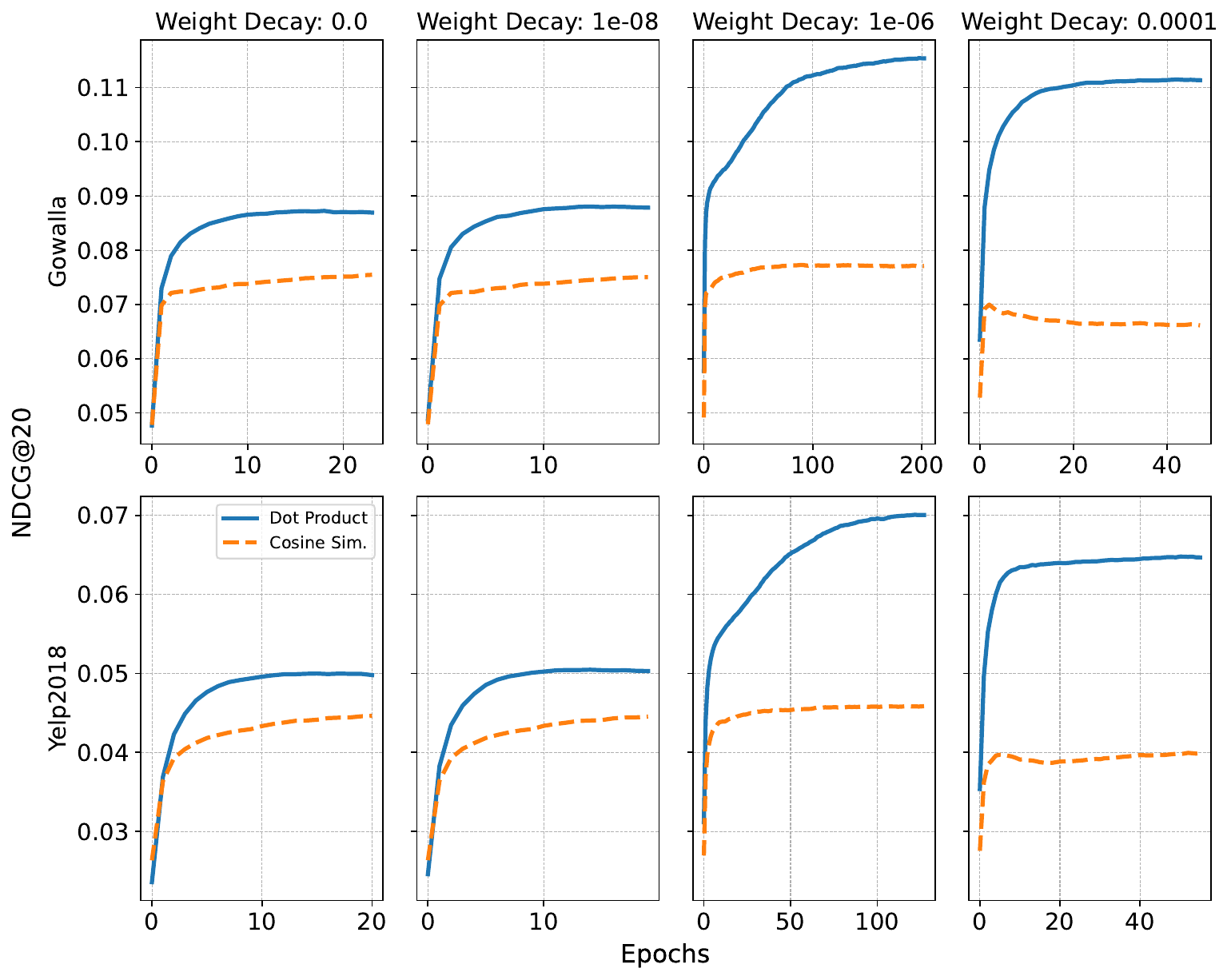}
    \caption{\textbf{Training Curves of DirectAU varying Weight Decay}. Cosine similarity is used to capture the angle information, and the dot product similarity is used to capture both angle and magnitude information. Given cosine similarity saturates early, a majority of training time is spent slowly encoding popularity information into the magnitudes via weight decay.}
    \label{fig:training_curves}
\end{figure}

\section{Theoretical Analysis}

In this section, we provide the proofs and analysis for the theory provided in \Cref{section:theory}. We begin with our derivation of \Cref{lemma:prob_batch}, and then provide the proof for \Cref{theorem:mag_change}. Afterwards, we provide supplemental analysis for the theorem.

\subsection{Derivation of Observation 1 }
\label{section:app_obs_1}
The goal of this observation is to give an equation to represent the asymmetric update process, particularly specifying the probability of a particular item receiving an in-batch vs. out-of-batch gradient update. 
We begin by denoting a batch of interactions as $B$, and the total number of interactions as $E$. We assume that the sampling process is uniform, and for simplicity, with replacement. While in practice the sampling is performed without replacement, given $|E| >> |B|$, the chance of repeatedly selecting a particular interaction is low. Then, the probability that a particular interaction $e \in E$ is sampled is $\frac{|B|}{|E|}$. We can then specify the compliment of this expression, and state the probability of a particular interaction not being in the batch as $1 - \frac{|B|}{|E|}$. 

We build on this expression and assess the probability that an item, with $d_{i}$ interactions, has at least one interaction in the batch. We begin by specify the probably that none of the interactions for the item are in the batch, which is expressed as: 

\[ P(i \notin B) = \prod^{d_{i}} \left(1 - \frac{|B|}{|E|}\right) = \left(1 - \frac{|B|}{|E|}\right)^{d_{i}}.\]
We then specify the compliment of this expression to express the probability that at least one interaction is within the sampled set. This is given as: 

\[ P(i \in B) = 1 - \left(1 - \frac{|B|}{|E|}\right)^{d_{i}}.\]
We will use this result in the next section to specify the probability that an embedding receives an in-batch vs. out-of-batch gradient update.  

\subsection{Proof of Theorem 1}
\label{section:theorem1_proof}

The goal in this analysis to understand what properties get induced into the user and item embeddings as a byproduct of the training loss updating batched users and items, while the weight decay updates all users and items. We will assume the chosen loss is the cosine similarity. This loss will only update pairs of vectors who share an edge, denoted by $E$. Specifically, the cosine similarity loss is: 

\[
f(\mathbf{u}, \mathbf{i}) = \frac{\mathbf{u} \cdot \mathbf{i}}{\|\mathbf{u}\| \|\mathbf{i}\|}
\]

From here, we can study what happens to the embedding as the training process unfolds. The full loss function can be specified, for a batch \( B \), as:
\[
L_{\text{full}} = -\left(\sum_{(u,i) \in B} 
 \frac{\mathbf{u} \cdot \mathbf{i}}{\|\mathbf{u}\| \|\mathbf{i}\|} \right)
 + \frac{\lambda}{2} \left( \sum_u \|\mathbf{u}\|^2 + \sum_i \|\mathbf{i}\|^2 \right)
\]

We then have the gradient descent (given we have a negative in front of the cosine similarity term) update for an item $i$, $\dfrac{\partial L_{\text{full}}}{\partial \mathbf{i}}$ can be expressed as:
$$
\dfrac{\partial L_{\text{full}}}{\partial \mathbf{i}} =
\begin{cases}
-\left(\frac{\mathbf{u}}{\|\mathbf{u}\| \|\mathbf{i}\|} - \frac{\mathbf{i} (\mathbf{u} \cdot \mathbf{i})}{\|\mathbf{u}\| \|\mathbf{i}\|^3}\right) + \lambda \mathbf{i}, (u, i) \in B \\
\lambda \mathbf{i}, i \notin B 
\end{cases}
$$
For brevity, we will assume that the interaction connected to item $i$ only appears once in the batch, which again is reasonable as the dataset size grows. This will lead to two corresponding gradient descent updates:
$$
\mathbf{i}^{(k+1)} =
\begin{cases}
\mathbf{i}^{(k)} + \eta\left(\frac{\mathbf{u}}{\|\mathbf{u}\| \|\mathbf{i}\|} - \frac{\mathbf{i} (\mathbf{u} \cdot \mathbf{i})}{\|\mathbf{u}\| \|\mathbf{i}\|^3}\right) - \eta\lambda \mathbf{i}, (u, i) \in B \\
\mathbf{i}^{(k)} - \eta \lambda \mathbf{i}, i \notin B 
\end{cases}
$$
These can be simplified to: 
$$
\mathbf{i}^{(k+1)} =
\begin{cases}
\mathbf{i}^{(k)} \left( 1 - \eta \lambda \right) + \eta\left(\frac{\mathbf{u}}{\|\mathbf{u}\| \|\mathbf{i}\|} - \frac{\mathbf{i} (\mathbf{u} \cdot \mathbf{i})}{\|\mathbf{u}\| \|\mathbf{i}\|^3}\right), (u, i) \in B \\
\mathbf{i}^{(k)} \left( 1 - \eta \lambda \right), i \notin B 
\end{cases}
$$
We can compute the magnitude change as a byproduct of these two possible gradient descent steps. 

$$
\|\mathbf{i}^{(k+1)}\|^{2} =
\begin{cases}
  \begin{aligned}[t]
    &\left(\mathbf{i}^{(k)} \left( 1 - \eta \lambda \right) + \eta \left(\frac{\mathbf{u}}{\|\mathbf{u}\| \|\mathbf{i}\|} - \frac{\mathbf{i} (\mathbf{u} \cdot \mathbf{i})}{\|\mathbf{u}\| \|\mathbf{i}\|^3}\right)\right) \\
    &\quad \cdot \left(\mathbf{i}^{(k)} \left( 1 - \eta \lambda \right) + \eta \left(\frac{\mathbf{u}}{\|\mathbf{u}\| \|\mathbf{i}\|} - \frac{\mathbf{i} (\mathbf{u} \cdot \mathbf{i})}{\|\mathbf{u}\| \|\mathbf{i}\|^3}\right)\right), \quad i \in B
  \end{aligned} \\
  \begin{aligned}[t]
    &\left(\mathbf{i}^{(k)} \left( 1 - \eta \lambda \right)\right) \cdot \left(\mathbf{i}^{(k)} \left( 1 - \eta \lambda \right)\right), \quad i \notin B
  \end{aligned}
\end{cases}
$$

These can be simplified to:
$$
\|\mathbf{i}^{(k+1)}\|^{2} =
\begin{cases}
    \| \mathbf{i}^{(k)}\|(1 -\eta\lambda)^{2} + \frac{\eta^{2}}{\| \mathbf{i}^{(k)}\|^{2}}\left(1 - \frac{(\mathbf{u} \cdot \mathbf{i})^{2}}{\|\mathbf{u}\|^{2}\|\mathbf{i}\|^{2}}\right) , \quad i \in B
    \\
    \| \mathbf{i}^{(k)}\|(1 -\eta\lambda)^{2}, \quad i \notin B
\end{cases}
$$
We can then specify the expected magnitude after the $k^{th}$ gradient descent update using the law of total expectation:
\begin{align*} \mathbb{E}[\|\mathbf{i}^{(k+1)}\|^{2}] &= P(i \in B) \mathbb{E}[\left(\|\mathbf{i}^{(k)}\|^2 (1 - \eta\lambda)^2 \right. + \left. \frac{\eta^{2}}{\|\mathbf{i}^{(k)}\|^{2}} \left(1 - \frac{(\mathbf{u} \cdot \mathbf{i})^2}{\|\mathbf{u}\|^2 \|\mathbf{i}\|^2}\right) \right)] \\ &\quad + P(i \notin B) \mathbb{E}[\|\mathbf{i}^{(k)}\|^2 (1 - \eta\lambda)^2] \end{align*}

By using our result from Observation 1, we can specify the probability $i \in B$ and $i \notin B$, leading to:
\begin{align*} \mathbb{E}[\|\mathbf{i}^{(k+1)}\|^{2}]&= \left(1 - \left(1 - \frac{|B|}{|E|}\right)^{d_{i}}\right) \\ &\quad \cdot \mathbb{E}[\left( \|\mathbf{i}^{(k)}\|^2 (1 - \eta\lambda)^2 \right. + \left. \frac{\eta^{2}}{\|\mathbf{i}^{(k)}\|^{2}} \left(1 - \frac{(\mathbf{u} \cdot \mathbf{i})^2}{\|\mathbf{u}\|^2 \|\mathbf{i}\|^2}\right) \right)] \\ &\quad + \left(1 - \frac{|B|}{|E|}\right)^{d_{i}}  \mathbb{E}[\|\mathbf{i}^{(k)}\|^2 (1 - \eta\lambda)^2]
\\ &\quad = \mathbb{E}[\|\mathbf{i}^{(k)}\|^2 (1 - \eta\lambda)^2] \\ &\quad + \left(1 - \left(1 - \frac{|B|}{|E|}\right)^{d_{i}}\right) \mathbb{E}[\left. \frac{\eta^{2}}{\|\mathbf{i}^{(k)}\|^{2}} \left(1 - \frac{(\mathbf{u} \cdot \mathbf{i})^2}{\|\mathbf{u}\|^2 \|\mathbf{i}\|^2}\right) \right)]
\end{align*}
Then,
\begin{align*} 
\mathbb{E}[\|\mathbf{i}^{(k+1)}\|^{2}] & = \mathbb{E}[\|\mathbf{i}^{(k)}\|^2]  + \mathbb{E}[\|\mathbf{i}^{(k)}\|^2 (\eta^2\lambda^2 - 2\eta\lambda)]\\ &\quad + \left(1 - \left(1 - \frac{|B|}{|E|}\right)^{d_{i}}\right) \mathbb{E}[\left. \frac{\eta^{2}}{\|\mathbf{i}^{(k)}\|^{2}} \left(1 - \frac{(\mathbf{u} \cdot \mathbf{i})^2}{\|\mathbf{u}\|^2 \|\mathbf{i}\|^2}\right) \right)]
\end{align*} 
We can then express the difference in expected magnitudes, $\mathbb{E}[\Delta \|\mathbf{i}^{(k+1)}\|^{2}]$ as:
\begin{align*} 
\mathbb{E}[\Delta \|\mathbf{i}^{(k+1)}\|^{2}]&=  \mathbb{E}[\|\mathbf{i}^{(k)}\|^2] \eta\lambda(\eta\lambda - 2) \\ &\quad + \left(1 - \left(1 - \frac{|B|}{|E|}\right)^{d_{i}}\right) \mathbb{E}[\left. \frac{\eta^{2}}{\|\mathbf{i}^{(k)}\|^{2}} \left(1 - \frac{(\mathbf{u} \cdot \mathbf{i})^2}{\|\mathbf{u}\|^2 \|\mathbf{i}\|^2}\right) \right)]
\end{align*} 
$\hfill \square $

\vspace{0.1cm}
\noindent \textbf{Analyzing the Magnitude Update Equation. } We begin by studying the first term in the equation, which is contributed by weight decay. For reasonable values of $\eta$ and $\lambda$, generally less than $1$, clearly the first term is negative. Thus, we can expect the weight decay term to decrease the magnitudes of the embeddings. 

The more interesting case is the second term. As the probability term is greater than or equal to $0$, we focus on the update term provided by cosine similarity. Specifically, we want to understand how this will change the embedding magnitudes. Let us first study when the contribution is greater than 0, then it is the case that:
\begin{align*}
1 \ge \left(\frac{(\mathbf{u} \cdot \mathbf{i}^{(k)})}{\|\mathbf{u}\| \|\mathbf{i}^{(k)}\|}\right)^{2}  &\rightarrow  \| \mathbf{i}^{(k)} \|^2 \|\mathbf{u}\|^2 \ge (\mathbf{u} \cdot \mathbf{i}^{(k)})^{2}
\\ &\quad = \| \mathbf{i}^{(k)} \|\|\mathbf{u}\| \ge (\mathbf{u} \cdot \mathbf{i}^{(k)}). 
\end{align*}
This statement is true for all vectors given the Cauchy-Schwarz inequality.
Then, let us assume study when the contribution is less than 0, constituting a decrease in magnitude. This would lead to the case:
\begin{align*}
1 < \left(\frac{(\mathbf{u} \cdot \mathbf{i}^{(k)})}{\|\mathbf{u}\| \|\mathbf{i}^{(k)}\|}\right)^{2} &\rightarrow   \| \mathbf{i}^{(k)} \| \|\mathbf{u}\| < (\mathbf{u} \cdot \mathbf{i}^{(k)}) 
\end{align*}
which violates the Cauchy-Schwarz inequality. Thus, we can conclude that the second term is strictly positive and will increase the embedding magnitude of $\mathbf{i}$

\vspace{0.1cm}
\noindent \textbf{Further Analysis.} In \Cref{fig:theory} we plot the expected magnitude change for typical values of weight decay, learning rate. Additionally, we fix the dot product term to be $0.9$ for explanatory purposes. Within the plot, the x-axis denotes the popularity or degree information in log scale, while the y-axis denotes the batch size, scaled relative to the full dataset size. The colors represent the magnitude change of the item embedding. We can see that for reasonable batch-sizes, typically in the range of $0.01$ to $0.05$, there is clear differentiation in magnitude updates as a function of popularity. Thus, at this level, we can expect strong encoding of popularity information. 

\begin{figure}[]
    \centering
    \includegraphics[width=8.3cm]{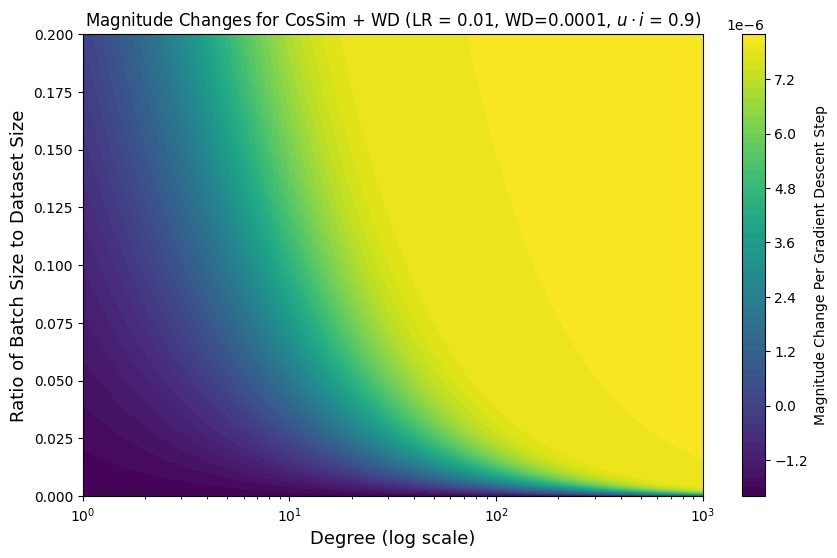}
    \caption{Theoretical magnitude changes for item embeddings as a function of popularity and batch size, relative to dataset size. At the extremes, such as when the batch size is close to 0, the update produces highly similar updates (either none, or similar positive magnitude). However, for values away from these extremes, we get strong correlation between degree and magnitude change, explaining our empirical results.}
    \label{fig:theory}
\end{figure}

\vspace{0.1cm}
\noindent \textbf{Discussion.} In practice, weight decay may not be applied to all users or items, such as in distributed settings. However, as long as there remains an asymmetry between the weight decay and loss updates, magnitude encoding will persist, as this asymmetry influences the probability term in \Cref{theorem:mag_change}. The specific probability for a distributed setting is dependent on implementation details, including how embedding tables are partitioned and the strategy used for applying weight decay. That said, in most popular deep learning frameworks, such as PyTorch, weight decay is applied to all parameters attached to the optimizer, not just those used in the loss computation. Thus, the asymmetry in updates is a default behavior in many setups. In the next section, we explore the implications of removing this asymmetry from the training process, as well as examine how various architectures affect the learning of magnitude and angle information.

\subsection{Proof of Corollary 1: \\ Batched vs. Full Table Weight Decay}
\label{corollary1}
In this section, we want to assess how the magnitude update changes for \textit{in-batch} weight decay updates, as opposed to full table updates. Using the cosine similarity loss function, we can express the gradient descent process for an item $i$ as,
$$
\mathbf{i}^{(k+1)} =
\begin{cases}
\mathbf{i}^{(k)} \left( 1 - \eta \lambda \right) + \eta\left(\frac{\mathbf{u}}{\|\mathbf{u}\| \|\mathbf{i}\|} - \frac{\mathbf{i} (\mathbf{u} \cdot \mathbf{i})}{\|\mathbf{u}\| \|\mathbf{i}\|^3}\right), (u, i) \in B \\
\mathbf{i}^{(k)}, i \notin B 
\end{cases}
$$
The update equations are nearly identical, but now the gradient is 0 for the instance the data point is not within the batch. The subsequent expected magnitudes after the $k^{th}$ gradient descent update step are,
\begin{align*} \mathbb{E}[\|\mathbf{i}^{(k+1)}\|^{2}] &= P(i \in B) \mathbb{E}[\left(\|\mathbf{i}^{(k)}\|^2 (1 - \eta\lambda)^2 \right. + \left. \frac{\eta^{2}}{\|\mathbf{i}^{(k)}\|^{2}} \left(1 - \frac{(\mathbf{u} \cdot \mathbf{i})^2}{\|\mathbf{u}\|^2 \|\mathbf{i}\|^2}\right) \right)] \\ &\quad + P(i \notin B) \mathbb{E}[\|\mathbf{i}^{(k)}\|^2] \end{align*}
Plugging in the probability of sampling a particular interaction within a batch and simplifying, we attain the change in magnitude update as,
\begin{align*} 
\mathbb{E}[\Delta \|\mathbf{i}^{(k+1)}\|^{2}]&=  \left(1 - \left(1 - \frac{|B|}{|E|}\right)^{d_{i}}\right)\mathbb{E}[\|\mathbf{i}^{(k)}\|^2] \eta\lambda(\eta\lambda - 2) \\ &\quad + \left(1 - \left(1 - \frac{|B|}{|E|}\right)^{d_{i}}\right)\mathbb{E}[\left. \frac{\eta^{2}}{\|\mathbf{i}^{(k)}\|^{2}} \left(1 - \frac{(\mathbf{u} \cdot \mathbf{i})^2}{\|\mathbf{u}\|^2 \|\mathbf{i}\|^2}\right) \right)]
\end{align*} 

\noindent  \textbf{Implications of Result:} When comparing to the case of full table updates, we can see that the second term is exactly the same, and thus the only difference is on the first term. Intuitively, having the probability exist on both terms means that the degree information can no longer create a discrepancy between the magnitude dampening that occurs by weight decay, and the magnitude increases incurred by the cosine similarity gradient. Instead, both updates always occur, and the changes in an item's magnitudes are  dependent the properties of the user and item embeddings themselves, i.e. the only thing that changes across epochs would be the dot product and magnitude terms. 

We note that is can still be possible that the magnitudes can systematically go up or down given the dot product and magnitude terms. Specifically, if an item continually has a low dot product score with an interacted user, the cosine similarity gradient term will have a large positive value. However, the key finding is that the popularity information has no means of being explicitly encoded as was the case in the full table updates. Empirically, this is verified in \Cref{fig:batched_wd} where the use of batched weight decay demonstrates an extremely weak correlation between popularity and degree, especially as compared to \Cref{fig:motivation_dau_bpr_mag_deg} where batched weight decay is employed. 

\subsection{Proof of Corollary 2: Negative Sampling and Weight Decay}
\label{corollary2}
In this section, we further generalize \Cref{theorem:mag_change} to incorporate negative sampling. Specifically, we assume that for each positive user-item pair within the batch, there are $\gamma$ non-interacted items sampled from the dataset. Across all positive user-item pairs, we assume that the negative samples are distinct and uniformly sampled, meaning there are $\gamma |B|$ total negative samples, and $(\gamma + 1)|B|$ total user-item pairs within the batch.

\subsubsection{Updating Probability of Ranking Update}
We first must update \Cref{lemma:prob_batch}, as now an item can either be sampled from the positive user-item interactions, or as a negative sample, meaning $P(i \in B) = P(i \in B_+ \cup i \in B_-)$, where $B_{+}$ are the positive interactions, and $B_{-}$ are the negative interactions. We know $P(i \in B_{+})$ from \Cref{lemma:prob_batch}, thus, we now need to find $P(i \in B_{-})$. First, negative samples are chosen randomly from the item set for each user. Thus, it is possible to choose from $|I|$ different items as negative samples, meaning for each user, an item may appear in the negative batch given the probability $P(i \in B_{-}) = \frac{\gamma}{|I|}$. Note that while excluding the positive samples from this set is ideal, most popular implementations use approximate negative sampling which simply uses the full item set, aligning with our setup.  Then, given our assumption that the negative samples across users in the batch are distinct, which is likely with a sufficiently large item set, the probability that an item is in the negative batch sample set across all users is given by $P(i \in B_{-}) = \frac{\gamma |B|}{|I|}$. Notably, while positive samples are dependent on the popularity of an item, the negative sampling does not depend on the popularity.

Given we did not condition on the specific user-item interaction pair when choosing items to be used as negative samples, we can view the generation of the positive and negative batches as independent events. Again, as long as the item set is large, the fraction of negative samples is much smaller than the item set size, these assumptions are reasonable. Then,
\begin{align*}
P(i \in B_{+} \cup i \in B_{-}) &= 1 - P(i \notin B_{+} \cap i \notin B_{-}) \\ &\quad   
 1 - (1 - P(i \in B_{+})(1-P(i \in B_{-})) \\ &\quad
 1-  \left(1 - \frac{|B|}{|E|}\right)^{d_{i}}\left(1-\frac{\gamma |B|}{|I|}\right)
\end{align*}
With this result, we can specify when an item, either from the positive or negative set, will experience a ranking loss update.

\subsubsection{Updated Ranking Loss}
\label{corollary3}
We can now study how positive and negative item embeddings are updated during training. The full loss with positive and negative samples can be specified, for a batch \( B = B_{+} \cup B_{-}\), as:

\begin{align*}
L_{\text{full}} &= -\left(\sum_{(u,i^{+}) \in B} 
 \frac{\mathbf{u} \cdot \mathbf{i^{+}}}{\|\mathbf{u}\| \|\mathbf{i^{+}}\|} \right) + \left(\sum_{(u,i^{-}) \in B} 
 \frac{\mathbf{u} \cdot \mathbf{i^{-}}}{\|\mathbf{u}\| \|\mathbf{i^{-}}\|} \right)
 \\ &\quad
 + \frac{\lambda}{2} \left( \sum_u \|\mathbf{u}\|^2 + \sum_i \|\mathbf{i}\|^2 \right)
\end{align*}
where $i^{+}$ denotes a positive sample, and $i^{-}$ denotes a negative sample. By minimizing this loss, positive pairs are optimized to have a larger cosine similarity, while negative pairs are optimized to have a smaller cosine similarity. We can now compute the partial derivative of $L_{full}$ with respect to $i^{+}$ and $i^{-}$, with the difference being the sign in front of the in-batch gradient.

$$
\dfrac{\partial L_{\text{full}}}{\partial \mathbf{i^{+}}} =
-\left(\frac{\mathbf{u}}{\|\mathbf{u}\| \|\mathbf{i^{+}}\|} - \frac{\mathbf{i^{+}} (\mathbf{u} \cdot \mathbf{i^{+}})}{\|\mathbf{u}\| \|\mathbf{i^{+}}\|^3}\right) + \lambda \mathbf{i^{+}}, (u, i^{+}) \in B 
$$

$$
\dfrac{\partial L_{\text{full}}}{\partial \mathbf{i^{-}}} =
\left(\frac{\mathbf{u}}{\|\mathbf{u}\| \|\mathbf{i^{-}}\|} - \frac{\mathbf{i^{-}} (\mathbf{u} \cdot \mathbf{i^{-}})}{\|\mathbf{u}\| \|\mathbf{i^{-}}\|^3}\right) + \lambda \mathbf{i^{-}}, (u, i^{-}) \in B 
$$

Otherwise, the gradient is just $\lambda \mathbf{i}$ when $i \notin B$. We can formulate the respective gradient descent steps as:

$$
\mathbf{i}^{+, (k+1)} =
\mathbf{i}^{+, (k)} \left( 1 - \eta \lambda \right) + \eta\left(\frac{\mathbf{u}}{\|\mathbf{u}\| \|\mathbf{i^{+}}\|} - \frac{\mathbf{i^{+}} (\mathbf{u} \cdot \mathbf{i^{+}})}{\|\mathbf{u}\| \|\mathbf{i^{+}}\|^3}\right), (u, i^{+}) \in B 
$$

$$
\mathbf{i}^{-, (k+1)} =
\mathbf{i}^{-, (k)} \left( 1 - \eta \lambda \right) - \eta\left(\frac{\mathbf{u}}{\|\mathbf{u}\| \|\mathbf{i^{-}}\|} - \frac{\mathbf{i^{-}} (\mathbf{u} \cdot \mathbf{i^{-}})}{\|\mathbf{u}\| \|\mathbf{i^{-}}\|^3}\right), (u, i^{-}) \in B 
$$

However, as shown in the proof of \Cref{theorem:mag_change}, when solving for the magnitude change, the cross term for the vector norm cancels out. Thus, despite $i$ and $i'$ having different signs in front of the cosine similarity gradient, the resulting magnitude updates are the same (just in opposite directions). Thus, we can combine the two update equations and specify the expected magnitude for an arbitrary item in the batch, $i$ that is either in the positive or negative set. Then, using the same simplification process in \Cref{section:theorem1_proof}, 

\begin{align*} 
\mathbb{E}[\Delta \|\mathbf{i}^{(k+1)}\|^{2}]&=  \mathbb{E}[\|\mathbf{i}^{(k)}\|^2] \eta\lambda(\eta\lambda - 2) + \\ &\quad 1-\left(  \left(1 - \frac{|B|}{|E|}\right)^{d_{i}}\left(1-\frac{\gamma |B|}{|I|}\right)\right) \mathbb{E}[\left. \frac{\eta^{2}}{\|\mathbf{i}^{(k)}\|^{2}} \left(1 - \frac{(\mathbf{u} \cdot \mathbf{i})^2}{\|\mathbf{u}\|^2 \|\mathbf{i}\|^2}\right) \right)]
\end{align*} 
Given our expression, it is clear the core difference is the inclusion of the negative sampling probability, which we can interpret as a modulation term between 0 and 1 on the original probability of \textit{not} being included in the batch. When $\gamma |B|$ is small relative to $|I|$, which is likely in large datasets, we should expect similar magnitude encoding as seen in the non-negative sampling case. However, as $\gamma |B|$ approaches $|I|$, the update tends to behave similar to that of the in-batch negative sampling where there is no discrepancy in coefficient between the weight decay and ranking loss updates (except now the coeffecient is 1).

\noindent \textbf{Implications of Result:} From a popularity perspective, negative sampling tends to reduce the encoding of popularity information as the negative sampling rate increases. This can be seen given that the out-of-batch probability term will go down, resulting on a higher overall weight on the cosine similarity term. Interestingly \cite{loveland2024understandingscalingcollaborativefiltering} draws a connection between the number of negative samples and the rank of the underlying user and item embedding matrices, while \cite{lin2024recommendationmodelsamplifypopularity} links the underlying rank of the user and item matrices to popularity bias. Thus, Corollary 2 offers a new perspective on how these relationships may come about in settings where weight decay is applied, specifically showing how negative sampling encourages a dampening of $P(i \in B)$, which directly induces popularity bias.

\subsection{Assessing Different Similarity Metrics}

In this section, we extend our analysis on cosine similarity to include two other similarity functions, the dot product and Euclidean distance. We follow a similar analysis to that in \Cref{section:theorem1_proof} to assess how these losses differ during optimization. Beginning with the dot product, the loss becomes, 

\[
L_{\text{full}} = -\left(\sum_{(u,i) \in B} 
 \mathbf{u} \cdot \mathbf{i} \right)
 + \frac{\lambda}{2} \left( \sum_u \|\mathbf{u}\|^2 + \sum_i \|\mathbf{i}\|^2 \right)
\]
for a batch $B$. The gradient for an in-batch and out-of-batch sample can then be expressed as,

$$
\dfrac{\partial L_{\text{full}}}{\partial \mathbf{i}} =
\begin{cases}
-\mathbf{u} + \lambda \mathbf{i}, (u, i) \in B \\
\lambda \mathbf{i}, i \notin B 
\end{cases}
$$
leading to gradient descent steps of,
$$
\mathbf{i}^{(k+1)} =
\begin{cases}
\mathbf{i}^{(k)} \left( 1 - \eta \lambda \right) + \eta\mathbf{u}, (u, i) \in B \\
\mathbf{i^{(k)}}(1 - \eta\lambda), i \notin B 
\end{cases}
$$
Then, the expected magnitude change can be expressed as,
\begin{align*} 
\mathbb{E}[\Delta\|\mathbf{i}^{(k+1)}\|^{2}] & = \mathbb{E}[\|\mathbf{i}^{(k)}\|^2 (\eta^2\lambda^2 - 2\eta\lambda)]\\ &\quad + \left(1 - \left(1 - \frac{|B|}{|E|}\right)^{d_{i}}\right) \mathbb{E}[2(\eta - \eta^{2}\lambda)(\mathbf{u}\cdot\mathbf{i}) + \eta^{2}\|\mathbf{u}\|^{2} ]
\end{align*} 
As the only change made in this setup was to the choice of ranking loss, the first term in the expected change is the same to the cosine similarity case of \Cref{section:theorem1_proof}. However, the change as a result of the ranking loss is significantly different. Most notably, the update as a result of the dot product gradient can no longer be bounded as a non-negative real number, thus, the magnitude change may either be positive or negative. While the amplification caused by the batch probability still exist, a certain behavior cannot be assessed for the dot product and is dataset, as well as initialization, dependent.

Next, we consider the Euclidean distance, as is used in DirectAU and MAWU of the main text. While the DirectAU and MAWU losses use normalized vectors when computing the loss, the gradient computation can still impact the magnitudes, similar to that of the dot product. We specify the updated loss as:

\[
L_{\text{full}} = \left(\sum_{(u,i) \in B} 
 \|\mathbf{u} - \mathbf{i}\|^{2} \right)
 + \frac{\lambda}{2} \left( \sum_u \|\mathbf{u}\|^2 + \sum_i \|\mathbf{i}\|^2 \right)
\]
The gradient for an in-batch and out-of-batch sample can then be expressed as,

$$
\dfrac{\partial L_{\text{full}}}{\partial \mathbf{i}} =
\begin{cases}
2(\mathbf{i} - \mathbf{u}) + \lambda \mathbf{i}, (u, i) \in B \\
\lambda \mathbf{i}, i \notin B 
\end{cases}
$$
leading to gradient descent steps of,
$$
\mathbf{i}^{(k+1)} =
\begin{cases}
\mathbf{i}^{(k)} \left( 1 - 2\eta - \eta \lambda \right) + 2\eta\mathbf{u}, (u, i) \in B \\
\mathbf{i^{(k)}}(1 - \eta\lambda), i \notin B 
\end{cases}
$$
Then, the respective magnitude changes are,

\[
\|\mathbf{i}^{(k+1)}\|^{2} =
\begin{cases}
\begin{aligned}[t]
&\bigl(1-\eta(1+\lambda)\bigr)^{2}\,\|\mathbf{i}^{(k)}\|^{2}
   + 2\,\eta\bigl(1-\eta(1+\lambda)\bigr)\,
     \mathbf{i}^{(k)}\!\cdot\!\mathbf{u}\\
&\quad + \eta^{2}\,\|\mathbf{u}\|^{2}, \qquad i \in B
\end{aligned}
\\[6pt]
\bigl(1-\eta\lambda\bigr)^{2}\,\|\mathbf{i}^{(k)}\|^{2}, \qquad i \notin B
\end{cases}
\]

From here, it becomes obvious that the core difference in the update equation and \Cref{theorem:mag_change} stems from the dot product between the item and user embedding, which can be either positive or negative. Thus, using the Euclidean norm introduces a mechanism to allow the gradient to possess sensitivity to magnitude information. This is in contrast to the cosine similarity result of \Cref{theorem:mag_change} where the update is strictly positive, regardless of the specific user or item embeddings.

\noindent \textbf{Implications of Result:} We use this result to explain the degradation observed when using SSM with PRISM for a non-linear MF model, as seen in \Cref{table:wd_vs_init}. Specifically, this result stems from the inherent nature of the SSM loss, where SSM uses cosine similarity, which is invariant to vector magnitudes and does not directly reward the preservation of the popularity signal that PRISM pre-encodes. When paired with non-linear MF with MLP layers, the transformation lacks the incentive to maintain magnitude differences; as a result, the pre-encoded popularity information is gradually washed out, leading to degraded performance. In contrast, DirectAU does not suffer from this issue because its alignment term, although computed on normalized vectors, leverages Euclidean distance during backpropagation and retains a dependence on the original magnitudes. This gradient reintroduces magnitude sensitivity, ensuring that the popularity signal from PRISM is preserved even in the presence of non-linear transformations. Thus, the observed degradation in the SSM case is not a flaw of PRISM itself but rather a consequence of the loss function and model architecture that fail to propagate magnitude information effectively.

\section{Additional Empirical Results}
\label{section:app_addtl_exp}
In this section we provide additional empirical results to support the claims in the main text.

\subsection{Results on LightGCN}
\label{section:app_lightgcn}

In this section we provide analysis on LightGCN, demonstrating the fact that message-passing directly encodes popularity information and does not greatly benefit from weight decay. In \Cref{fig:lightgcn_app} we show performance across weight decay levels for Gowalla and Yelp2018, trained with DirectAU. While there are subtle changes in the popular and unpopular performance metrics, these are significantly less pronounced than the changes we see for DirectAU applied to MF. Thus, from a practical perspective, we conclude that LightGCN does not need weight decay when trained. 

\begin{figure}[] 
    \centering \includegraphics[width=8.5cm]{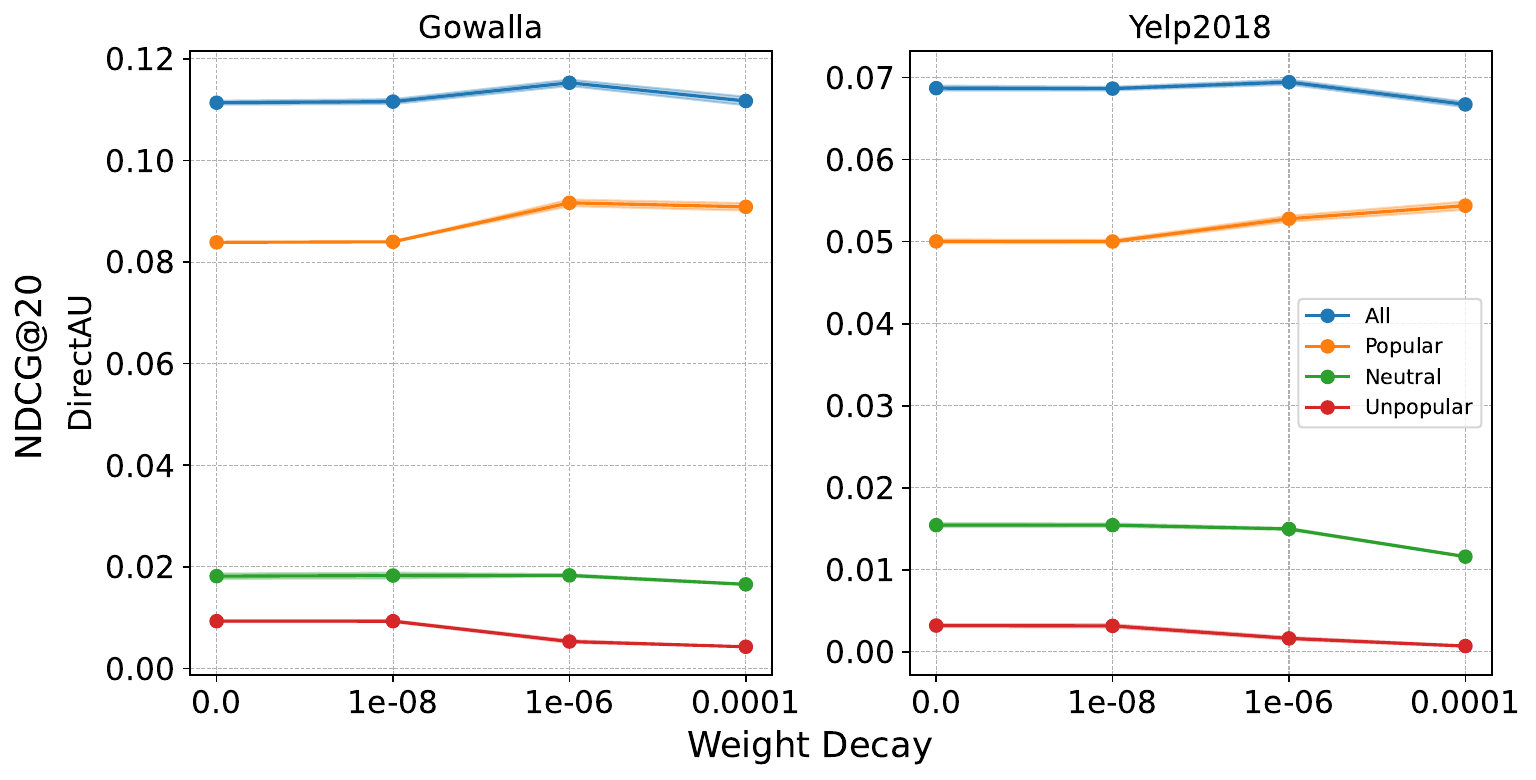} 
    \caption{Results for LightGCN trained with DirectAU. Across weight decay levels, we see that performance is highly consistent, indicating that LightGCN does not greatly benefit from weight decay.} 
\label{fig:lightgcn_app} 
\end{figure}

\subsection{Results on Debiasing}

To demonstrate the effectiveness of PRISM's debiasing capabilities, we compare it to three strong benchmarks including ReSN (spectral regularization to dampen popularity information), IPL \cite{liu2023ipl} (direct regularization on popularity info), and MACR \cite{Wei_2021}(causal debiasing). For ReSN and IPL, the methods are applied to BPR as recommended in their paper. For MACR, the BCE-based loss is utilized, again as recommended. Hyperparameters are searched over the ranges provided in the original paper. In \Cref{fig:debias}, we plot the ratio of unpopular to popular NDCG@20 versus the overall NDCG@20 performance. The ratio of unpopular to popular NDCG@20 metrics measures the relative contribution of each on the overall performance, with a higher ratio denoting a more fair model.

\label{section:app_debias}
\begin{figure}[] 
    \centering \includegraphics[width=8.5cm]{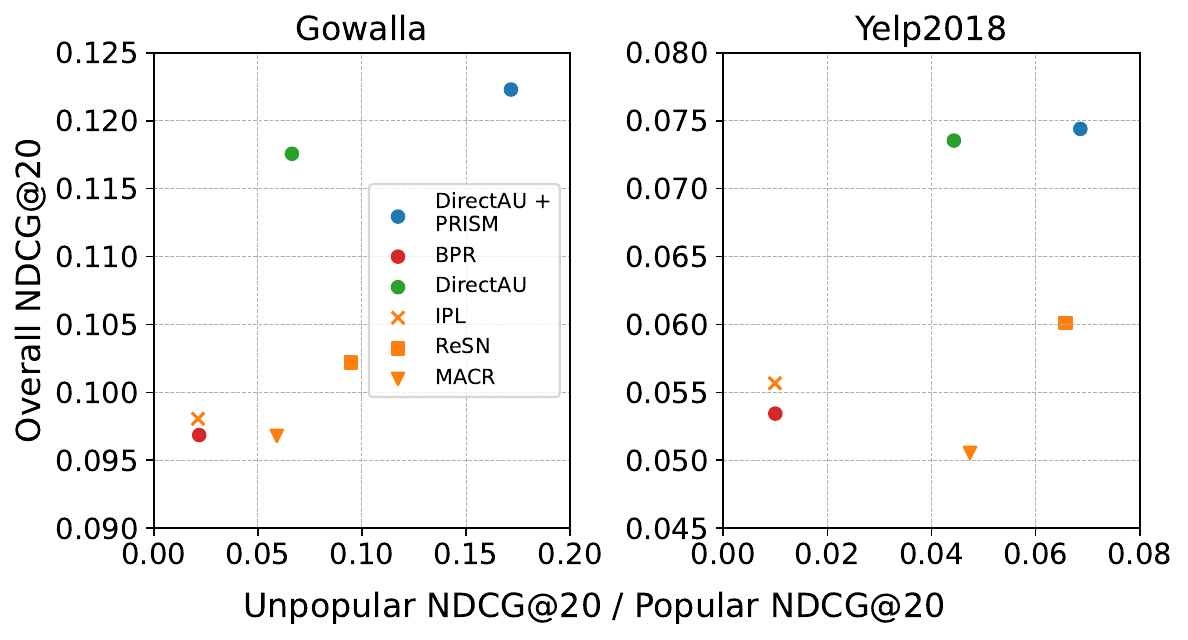} 
    \caption{Results for BPR, DirectAU, Prism + DirectAU, and De-biasing baselines. Unpopular NDCG@20 / Popular NDCG@20 captures how well the model mitigates popularity bias, and thus higher is better. PRISM + DirectAU produces the best overall model in both performance and debiasing. } 
\label{fig:debias} 
\end{figure}
In \Cref{fig:debias}, we show these two results for BPR, DirectAU, PRSIM with DirectAU, and the three de-biasing baselines. Specifically, models which are near the top of the plot are higher performing, while models to the right experience less popularity bias. While we can see that the de-biasing baselines indeed perform better than standard BPR, PRISM produces significantly better unpopular to popular performance ratios, indicating better de-biasing. We note that better performance comes from the use of DirectAU, however, it is with PRISM that the de-biasing is achieved.

\subsection{Results on SSM and MAWU for MF}

To accompany the results in \Cref{fig:motivation_dau_bpr_wd}, we provide similar results for SSM and MAWU. As seen in \Cref{fig:ssm_mawu_wd}, the trends are highly similar to the results of BPR and DirectAU, demonstrating the wide-spread behavior of weight decay. 

\begin{figure*}[] 
    \centering \includegraphics[width=16.cm]{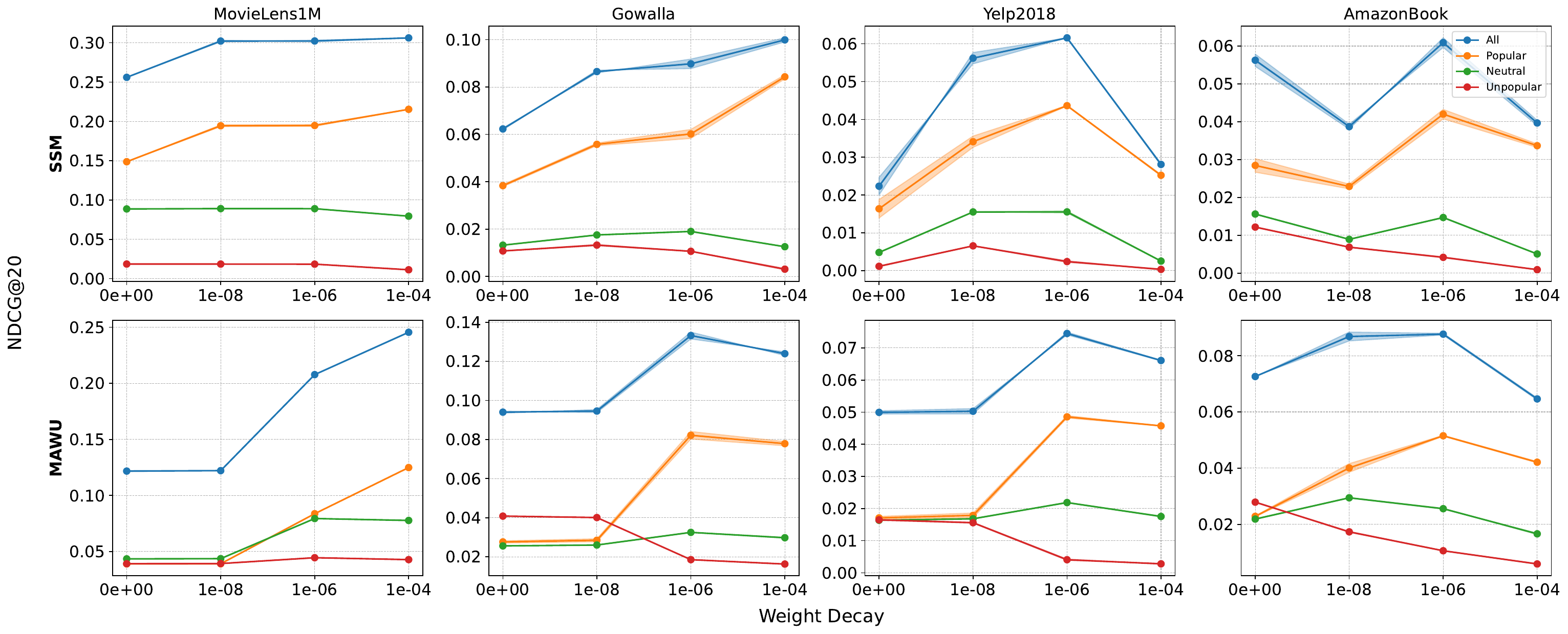} 
    \caption{\textbf{Performance with Varying Weight Decay: SSM (top) and MAWU (bottom)}. Increasing the strength of weight decay similarly demonstrates a preference towards items with high popularity. } 
\label{fig:ssm_mawu_wd} 
\end{figure*}

\subsection{Results on Non-linear MF}
\label{section:app_nonlin_mf}

To demonstrate that the weight decay results hold on neural-style architectures, we provide similar results as seen in \Cref{fig:motivation_dau_bpr_wd} to non-linear MF. As seen in \Cref{fig:dim2}, the trends are again highly similar across losses, where weight decay displays a strong correlation to the gap between popular and unpopular performance. Interestingly, non-linear MF's utilization of MLPs has the capacity to modulate the embedding magnitudes during message passing, which could potential to disrupt the popularity encoding. Given these empirical results, we argue that this effect is not as significant, and these models are still susceptible to weight decay's popularity encoding. 

\begin{figure*}[] 
    \centering \includegraphics[width=16.cm]{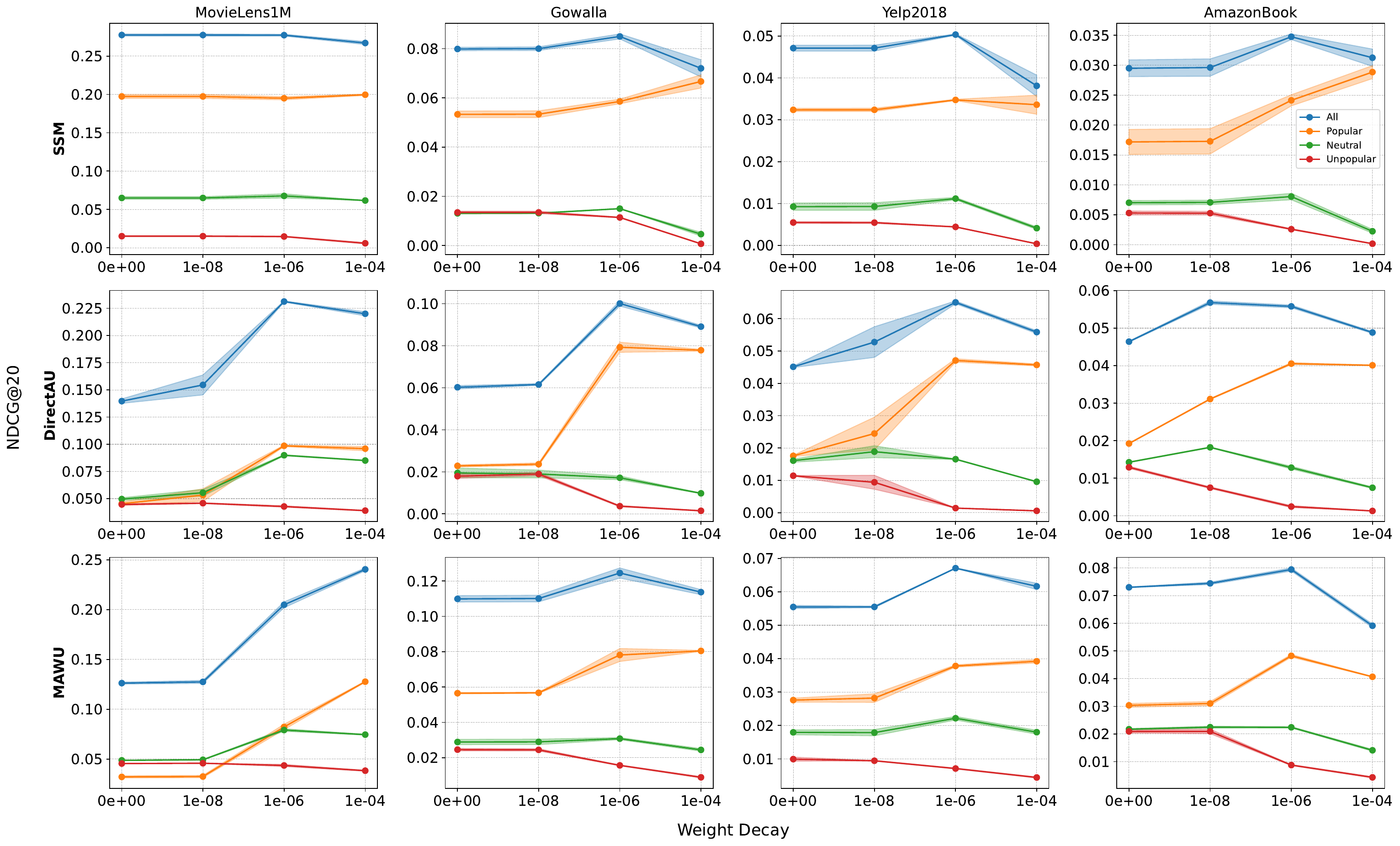} 
    \caption{{Performance with Varying Weight Decay for Non-linear MF: SSM (top), DirectAU (middle), and MAWU (bottom)}. Across the losses, we see similar increases in popularity performance as a function of weight decay. } 
\label{fig:dim2} 
\end{figure*}

\end{document}